\documentclass{new_tlp}
\usepackage{mathptmx}

\title[Theoretical Study of Tabled CLP]{%
  A Theoretical Study of (Full) Tabled Constraint Logic Programming
  \thanks{
    Work partially supported by EIT Digital, MINECO
    project TIN2015-67522-C3-1-R (TRACES), MICINN project
    PID2019-108528RB-C21 (ProCode), and Comunidad de Madrid project
    S2018/TCS-4339 BLOQUES-CM co-funded by EIE Funds of the European Union.
  }
}

\author[Joaqu\'in Arias and Manuel Carro]{%
Joaqu\'in Arias$^{1,2}$ \and Manuel Carro$^{1,3}$ \\
$^1$IMDEA Software Institute, $^2$Universidad Rey Juan Carlos, $^3$Universidad Polit\'ecnica de Madrid\\
\texttt{\emph{joaquin.arias@urjc.es, manuel.carro@\{imdea.org,upm.es\}}} 
}

\jdate{March 2003}
\pubyear{2003}
\pagerange{\pageref{firstpage}--\pageref{lastpage}}
\doi{S1471068401001193}

\usepackage{url} 
\usepackage[hidelinks]{hyperref}
\usepackage[ruled,lined]{algorithm2e}
\usepackage{graphicx}

\usepackage[utf8]{inputenc} 
\usepackage{booktabs}

\usepackage{amssymb}
\usepackage{amsmath}
\usepackage{textcomp}
\usepackage{calc}
\usepackage{comment}
\usepackage[justification=centering]{caption}
\usepackage[list=true]{subcaption}

\usepackage{relsize}

\usepackage{mathptmx}
\usepackage{dsfont}

\usepackage{calrsfs}
\DeclareMathAlphabet{\pazocal}{OMS}{zplm}{m}{n}
\newcommand{\domain}{\ensuremath{\pazocal{D}}\xspace}
\newcommand{\mf}{\ensuremath{\pazocal{F}}\xspace}

\usepackage{xspace}  
\newcommand{\R}{$\mathds{R}$\xspace}
\newcommand{\Q}{$\mathds{Q}$\xspace}
\newcommand{\h}{$\mathds{H}$\xspace}

\newcommand{\lfp}{\ensuremath{\mathit{lfp}}\xspace}

\usepackage{enumitem}
\setlist[enumerate,2]{ref=\arabic{enumi}.\alph*}

\usepackage{multirow}
\usepackage{longtable}

\usepackage{adjustbox}
\usepackage{rotating}

\usepackage{relsize}
\newcommand{\bigsymbol}[1]{\mathlarger{\mathlarger{#1}}}

\usepackage{listings} 
\usepackage{textcomp}
\newcommand{\code}{\lstinline[style=MyInline]}

\lstnewenvironment{tablstcode}
{\lstset{style=MyProlog,numbers=none,xleftmargin=0cm,belowskip=-0.8 \baselineskip,aboveskip=-0.6 \baselineskip}}
{}
\lstdefinestyle{tree}
{
  basicstyle = \small\ttfamily\color{PrologPredicate},
  basewidth = 0.5em,
  moredelim = {[s][\color{PrologString}]{ \{}{\} }},
  moredelim = {*[s][{\color{PrologVar}}]{(}{)}},
  literate     =
  {.\\=.}{{\ \char"5C=\ }}3
  {\\=}{{\ \char"5C=\ }}3
  {.<.}{{\ \#<\ }}4
  {.>.}{{\ \#>\ }}4
  {=}{{\ =\ }}3
  {.=.}{{\ \#=\ }}3
  {.=<.}{{\ \#=<\ }}5
  {.>=.}{{\ \#>=\ }}5
}
\lstdefinestyle{MyInline}
{
  basicstyle = \ttfamily\color{PrologOther},
  breaklines = true,
  breakatwhitespace=true,
  upquote = true,
  literate =
  {,}{}{1\discretionary{,}{}{,}}
  {|}{{\color{PrologOther}$\mid$}}1
  {\\\{}{{\color{PrologOther}\{}}1
  {\\\}}{{\color{PrologOther}\}}}1
  {[}{{\color{PrologOther}\small[}}1
  {]}{{\color{PrologOther}\small]}}1
  {\\$}{{\$}}1
  {.=.}{{\color{PrologOther}\,\#=\,}}3
  {.<.}{{\color{PrologOther}\,\#<\,}}3
  {.>.}{{\color{PrologOther}\,\#>\,}}3
  {.=<.}{{\color{PrologOther}\,\#=<\,}}4
  {.>=.}{{\color{PrologOther}\,\#>=\,}}4
  {==}{{\color{PrologOther}\,==\,}}3
  {<}{{\color{PrologOther}\,<\,}}2
  {>}{{\color{PrologOther}\,>\,}}2
  {=<}{{\color{PrologOther}\,=<\,}}3
  {>=}{{\color{PrologOther}\,>=\,}}3
  {.\\=.}{{\color{PrologOther}{\char"5C}=}}3
  {\\=}{{\color{PrologOther}\char"5C=}}3
  {?-}{{\color{black}\;?\vspace{-.5em}-\,}}4
  {:-}{{\color{black}:\vspace{-.5em}-\,}}3
}
\lstdefinestyle{MyProlog}
{
  keywords = {},
  upquote = true,
  basicstyle = \relsize{-1}\ttfamily\color{PrologPredicate},
  basewidth = 0.48em,
  moredelim = {**[s][\color{PrologString}]{'}{'}},
  moredelim = {**[is][\color{PrologString}]{"}{"}},
  moredelim = {**[is][\color{PrologComment}]{`}{`}},
  moredelim = {**[is][\color{PrologPredicate}]{@}{@}},
  moredelim = {*[s][\color{PrologVar}]{(}{)}},
  moredelim = {*[s][\color{PrologOther}]{:-}{.}},
  moredelim = {*[s][\color{red}]{/*}{*/}},
  commentstyle = \mdseries\color{PrologComment},
  morecomment=[l]\%,
  morecomment=[s]{/*}{*/},
  literate     =
  {[}{{\color{PrologOther}\small[}}1
  {]}{{\color{PrologOther}\small]}}1
  {\\$}{{\$}}1
  {&(}{{\color{PrologOther}(}}1
  {&)}{{\color{PrologOther})}}1
  {&.}{{.}}0
  {.=.}{{\color{PrologOther}\ \#=\ }}4
  {.<.}{{\color{PrologOther}\ \#<\ }}4
  {.>.}{{\color{PrologOther}\ \#>\ }}4
  {.=<.}{{\color{PrologOther}\ \#=<\ }}5
  {.>=.}{{\color{PrologOther}\ \#>=\ }}5
  {.\\=.}{{\color{PrologOther}\char"5C=}}3
  {\\=}{{\color{PrologOther}\char"5C=}}3
  {,}{{\color{PrologOther}\footnotesize,}}1
  {;}{{\color{PrologOther}\footnotesize;}}1,
}
\lstdefinestyle{MyASP}
{
  keywords = {},
  upquote = true,
  basicstyle = \relsize{-1}\ttfamily\color{PrologPredicate},
  basewidth = 0.52em,
  moredelim = {*[s][\color{PrologOther}]{:-}{.}},
  moredelim = {*[s][\color{PrologOther}]{\{}{\}}},
  moredelim = {[s][\color{PrologVar}]{(}{)}},
  commentstyle = \color{PrologComment},
  morecomment=[l]\%,
  morecomment=[l]?,
  literate     =
  {..}{..}2
  {&(}{{\color{PrologOther}(}}1
  {&)}{{\color{PrologOther})}}1
  {,}{{\footnotesize,}}1
}
\usepackage{xcolor}
\definecolor{PrologPredicate}{RGB}{0,0,200}
\definecolor{PrologVar}      {RGB}{145,032,039}
\definecolor{PrologComment}  {RGB}{169,082,044}
\definecolor{PrologOther}    {rgb}{0.1,0.1,0.1}
\definecolor{PrologString}   {RGB}{070,120,200}

\usepackage{tikz}
\usetikzlibrary{shapes.geometric, arrows, patterns}

\usepackage[edges]{forest}

\usepackage{morefloats}
\usepackage{xcolor}
\usepackage[textwidth=0.15\textwidth,textsize=footnotesize]{todonotes}
\newcommand{\ok}{{\color{PrologPredicate}{\checkmark}}\xspace}    

\newcommand{\no}{{\color{PrologVar}{\ensuremath{\boldsymbol{\times}}}}\xspace}

\let\proof\relax

\usepackage{amsthm}
\newtheorem{deff}{Definition}
\newtheorem{lemm}{Lemma}
\newtheorem{thee}{Theorem}
\newtheorem{exaa}{Example}
\newtheorem{corr}{Corollary}

\let\oldc\,
\renewcommand{\,}{\oldc \allowbreak}
\let\oldland\land
\renewcommand{\land}{\oldland \allowbreak}

\newcommand{\goalset}[2]{%
  \ensuremath{\langle #1,\, #2 \rangle}%
}
\newcommand{\goaltwo}[2]{%
  \ensuremath{\langle \{#1\},\, #2 \rangle}%
}
\newcommand{\goal}[1]{%
  \ensuremath{\langle #1 \rangle}%
}

\newcommand{\calltwo}[2]{\ensuremath{( #1,\, #2 )}}
\newcommand{\callcode}[2]{$($\code|#1|, \code|#2|$)$}

\newcommand{\mt}[1]{\ensuremath{\mathtt{#1}}}
\newcommand{\newstate}[3]{
  \ifthenelse{\equal{#2}{\,}}
  { \textbf{#1} $\langle\ \emptyset\ $, \ensuremath{\thickmuskip=0mu \medmuskip=2mu \mathtt{#3}}$\rangle$
    }
  { \textbf{#1} $\langle\{$\texttt{#2}$\}$, \ensuremath{\thickmuskip=0mu \medmuskip=2mu \mathtt{#3}}$\rangle$ }
}

\newcommand{\newstatelarge}[3]{
  { \textbf{#1} $\langle\{$\texttt{#2}$\}$, \\ \hspace{2em} \ensuremath{\thickmuskip=0mu \mathtt{#3}}$\rangle$ }
}
\newcommand{\Cs}{\ensuremath{C_{g}}\xspace}
\newcommand{\As}{\ensuremath{A_{\goaltwo{g}{c_g}}}\xspace}

\newcommand*{\pr}[1]{%
  \ifcase#1%
  \relax
  \or \ensuremath{\mathtt{_1}}\nobreak
  \or \ensuremath{\mathtt{_2}}\nobreak
  \or \ensuremath{\mathtt{_3}}\nobreak
  \or \ensuremath{\mathtt{_4}}\nobreak
  \or \ensuremath{\mathtt{_5}}\nobreak
  \else
  \fi
}

\begin{document}
\maketitle

  \hyphenpenalty 10000 %
  \exhyphenpenalty 10000 %

\begin{abstract}
  Logic programming with tabling and constraints (TCLP, \emph{tabled
    constraint logic programming}) has been shown to be more
  expressive and, in some cases, more efficient than LP, CLP, or LP with
  tabling.
%
%
  %
  In this paper we provide insights regarding the semantics,
  correctness, completeness, and termination of top-down execution
  strategies for full TCLP, i.e., TCLP featuring entailment checking
  in the calls and in the answers.  
  We present a top-down semantics for TCLP and show that it is equivalent
  to a fixpoint semantics.
  We study how the constraints that a program generates can
  effectively impact termination, even for constraint classes that are
  not constraint compact, generalizing previous results.
%
  We also present how different variants of constraint projection
  impact the correctness and completeness of TCLP implementations.
  %
  All of the presented characteristics are implemented (or can be
  experimented with) in Mod~TCLP, a modular framework for Tabled
  Constraint Logic Programming, part of the Ciao Prolog logic
  programming system. 
\end{abstract}


\section{Introduction and Motivation}
\label{sec:introduction}

Constraint Logic Programming (CLP)~\cite{survey94} extends
Logic Programming (LP)
with variables that can belong to arbitrary constraint domains and
the ability to incrementally solve equations involving these
variables.
CLP brings additional expressive power to LP, since constraints can
very concisely capture complex relationships.  Also,
shifting from ``generate-and-test'' to ``constraint-and-generate''
patterns reduces the search tree and therefore brings additional
performance, even if constraint solving is in general more expensive
than first-order unification.

Tabling~\cite{tamaki.iclp86,Warren92} is an execution
strategy for logic programs that suspends repeated calls which could
cause infinite loops.  Answers from non-looping branches are
used to resume suspended calls which can, in turn, generate more answers.
Only new answers are saved, and evaluation finishes when no new
answers can be generated.  Tabled evaluation always terminates for
calls/programs with the bounded term depth property (those that can only
generate terms with a finite bound on their depth) and can improve
efficiency for terminating programs that repeat computations, as it automatically
implements a variant of dynamic programming.
%
%
Tabling has been successfully applied in a variety of contexts, including
deductive databases, program analysis, semantic Web reasoning, and
model
checking~\cite{pracabsin,Dawson:pldi96,zou05:owl-tabling,ramakrishna97:model_checking_tabling,Charatonik02:model_checking_tabulation}.


%


The integration of tabling and constraint solving, Tabled Constraint
Logic Programming (TCLP), makes it possible to exploit their synergy
in several application fields of which we highlight a few:

\begin{description}[noitemsep]
\item [Abstract interpretation:] Tabling can be used naturally to
  compute
  fixpoints~\cite{kanamori93:absint-oldt,janssens98:tabling_for_AI-tapd},
  but, additionally, by implementing abstract domain operations as
  constraints~\cite{arias19:ciaopp-tclp}, entailment will
  automatically detect more particular calls and suspend their
  execution to reuse analysis results from most general calls, thereby
  speeding up the fixpoint computation.  Constraints can also be used
  to state \emph{preconditions} to the analysis results before the
  analysis starts in a powerful yet flexible fashion.  These
  preconditions can propagate during the evaluation and help solve
  some verification problems faster.

\item [Reasoning on ontologies:] An ontology formalizes types,
  properties, and interrelationships among entities.  They can be
  expressed as a lattice constraint system and, with TCLP, evaluation in
  ontologies can benefit from entailment of instances which are more
  particular than other entities, in a fashion similar to OWL
  (\url{www.w3.org/owl}), but in potentially richer domains and/or
  more complex scenarios (e.g., stream data
  analysis~\cite{arias-dc-iclp2016}).

\item [Constraint-based verification:] Verification conditions can be
  encoded as constraint systems, and the tabling engine can use
  entailment to guarantee termination and save execution
  time~\cite{Charatonik02:model_checking_tabulation,jaffar04:timed_automata,navas2013-FTCLP}.

\item [Incremental evaluation of aggregates:] For aggregates that can
  be embedded into a lattice (e.g., minimum), the aggregation
  operation can be expressed based on the partial order of the
  lattice. In these cases, the aggregate operations in the lattice can
  be seen as a counterpart of the operations among constraints defined
  in TCLP~\cite{atclp-padl2019}.
\end{description}

\begin{figure}
  \centering
  
  \begin{subfigure}[b]{.33\textwidth}
\begin{lstlisting}[style=MyProlog]
dist(X, Y, D) :- 
    dist(X, Z, D1), 
    edge(Z, Y, D2), 
    D is D1 + D2.
dist(X, Y, D) :-
    edge(X, Y, D).

?- dist(a,Y,D), D<K.
\end{lstlisting}
    \caption{LP version.}
  \label{fig:dist-lp}
  \end{subfigure}
  \begin{subfigure}[b]{.33\textwidth}
\begin{lstlisting}[style=MyProlog]
dist(X, Y, D) :- 
    D1.>.0, D2.>.0, 
    D.=.D1+D2, 
    dist(X, Z, D1), 
    edge(Z, Y, D2).
dist(X, Y, D) :- 
    edge(X, Y, D).

?- D.<.K, dist(a,Y,D).
\end{lstlisting}
    \caption{CLP(\R) version.}
  \label{fig:dist-clp}
  \end{subfigure}
  \begin{subfigure}[b]{.32\textwidth}
\begin{lstlisting}[style=MyProlog]
dist(X, Y, D) :- 
    D1.>.0, D2.>.0, 
    D.=.D1+D2, 
    edge(X, Z, D1), 
    dist(Z, Y, D2).
dist(X, Y, D) :- 
    edge(X, Y, D).

?- D.<.K, dist(a,Y,D).
\end{lstlisting}
    \caption{Right-recursive CLP(\R) version.}
  \label{fig:dist-clp-right}
  \end{subfigure}
  
  \caption[Distance traversal in a graph.]{Distance traversal in a
    graph.\\Note: The symbols \code|.>.| and \code|.=.| are
    (in)equalities in CLP(\R).}

\end{figure}

In order to highlight some of the advantages of TCLP vs.\ LP, tabling,
and CLP with respect to declarativeness and logical reading,
in~\cite{TCLP-tplp2019} we compared how different versions of a
program to compute distances between nodes in a graph behave under
these three approaches. Each version was adapted to a different
paradigm, but trying to stay as close as possible to the original
code, so that the additional expressiveness can be solely attributed to the
evaluation strategy rather than to differences in
the code itself. Their behaviors are summarized in
Table~\ref{tab:distance-summary} and explained below:

\begin{itemize}
\item {\bf LP}: 
  The code in Fig.~\ref{fig:dist-lp} is the Prolog version of a
  program used to find the distance between two nodes in a graph.  The
  distance between two nodes\footnote{This is a typical query for the analysis
    of social networks~\cite{swift2010:subsumption}.}
 is calculated by adding variables
  \code{D1} and \code{D2}, corresponding to distances to and from an
  intermediate node, once they are instantiated.
  %
  The figure also shows a query used to determine which node(s)
  \code{Y} is/are within a distance \code{K} from node \code{a}.
  This query does not terminate as left recursion makes the recursive
  clause enter an infinite loop.
  If we convert the program to a right-recursive version by swapping
  the calls to \code{edge/3} and \code{dist/3}, the program will still
  not terminate in a cyclic graph.
  
\item {\bf CLP(\R)}:
  Fig.~\ref{fig:dist-clp} is the CLP(\R) version of the same code where
  addition is modeled as a constraint and placed at the beginning of
  the clause.  Since the total distance \code{D} is bound by the
  constraint \code{D #< K} in the query, the search would be expected
  to be pruned if \code{D} exceeds the maximum distance, \code{K}.
  However, the constraints placed before the recursive call do not
  cause this bound to be violated, and therefore it would enter a loop
  even for graphs without loops.
  The right-recursive version of the CLP(\R) program in
  Fig.~\ref{fig:dist-clp-right} will however finish because the
  initial bound to the distance eventually causes the constraint store
  to become inconsistent, which provokes a failure in the search.
  Note that this transformation is easy in this case, but it would not
  have the same effect should the clause be written with a
  (logically equivalent) double recursion. This is optional in this
  example, but it may be necessary or more natural in other cases,
  such as in parsing applications, language interpreters, algorithms
  on trees, or divide-and-conquer algorithms.


\item {\bf Tabling}:
  Tabling records the first occurrence of each call to a tabled
  predicate (the \emph{generator}) and its answers.
  In variant tabling, the most usual form of tabling, when a call
  equal up to variable renaming to a previous generator is found (a
  variant), its execution is suspended, and it is marked as a
  \emph{consumer} of the generator.  For example, \texttt{dist(a,Y,D)}
  is a variant of \texttt{dist(a,Z,D)} if \texttt{Y} and \texttt{Z}
  are free variables.
  When a generator finitely finishes exploring all of its clauses and
  its answers are collected, its consumers are resumed and are fed the
  answers of the generator.  This may make consumers produce new
  answers that will in turn cause more resumptions.
  Tabling is a complete strategy for all programs with the bounded
  term-depth property, which in turn implies that the Herbrand model
  is finite.  Therefore, left- or right-recursive \emph{reachability}
  terminates in finite graphs with or without cycles.  However, the
  program in Fig.~\ref{fig:dist-lp} has an infinite minimum Herbrand
  model for cyclic graphs: every cycle can be traversed an unbound
  number of times, giving rise to an unlimited number of answers with
  a different distance each.  The query \code{?-dist(a,Y,D),D<K} will
  therefore not terminate under variant tabling.

\item {\bf TCLP}:
  The program in Fig.~\ref{fig:dist-clp} can be executed with tabling
  and using constraint entailment to suspend calls which are \emph{more
  particular} than previous calls and, symmetrically, to keep only the
  most general answers returned.
  Entailment can be seen as a generalization of subsumption for the
  case of general constraints; in turn, subsumption was shown to
  enhance termination and performance in
  tabling~\cite{swift2010:subsumption}.
  When a goal $G_1$ entails another goal $G_0$, the solutions for
  $G_1$ are a subset of the solutions for $G_0$.
  To make the entailment relationship explicit, we define a TCLP goal
  as \mbox{\calltwo{g}{c_{g}}} where $g$ is the call (a literal) and
  $c_g$ is the projection of the current constraint store onto the
  variables of the call.
  Then, a goal $G_0 = $\callcode{dist(X,Y,D)}{D<150} is entailed by
  another goal $G_1 = $\callcode{dist(X,Y,D)}{D>0$\ \land\ $D<75} because
  the solutions for \code{D>0$\ \land\ $D<75} are contained in the
  solutions for \code{D<150} %
  (\code{D>0$\ \land\ $D<75$\ \sqsubseteq\ $D<150}), and we write $G_1
  \sqsubseteq G_0$.
  We say that $G_1$, the more particular goal, is the \emph{consumer},
  and $G_0$, the most general goal, is the \emph{generator}.
  The key observation behind the use of entailment in TCLP is that 
 calls to more particular goals can suspend their execution and
  later recover the answers collected by the most general call and
  continue execution.
  %
  %
  %
  The solutions for the consumer are a subset of that for the
  generator.
  %
  However, some answers for a generator may not be valid for a
  consumer.  For example, %
  \code{D>125$\ \land\ $D<135} is a solution for $G_0$
  but not for $G_1$, since $G_1$ has a 
  constraint store more restrictive than the $G_0$.  Therefore, the
  tabling engine should check and filter, via the constraint solver,
  that answers from generators are consistent with the constraint
  store of consumers.

\end{itemize}

\begin{table}
  \caption{Termination properties comparison of LP, CLP, tabling and TCLP.}
  \label{tab:distance-summary}
  \centering
  \begin{tabular}{llcccc}
    \toprule
    Graph          &                      & LP     & CLP    & TAB   & TCLP  \\
    \midrule
    Without cycles & Left recursion       & \no    & \no    & \ok   & \ok   \\
                   & Right recursion      & \ok    & \ok    & \ok   & \ok   \\   
    \noalign{\vspace {.25cm}}
    With cycles    & Left recursion       & \no    & \no    & \no   & \ok   \\ 
                   & Right recursion      & \no    & \ok    & \no   & \ok   \\
    \bottomrule
  \end{tabular}
\end{table}

The use of entailment in calls and answers enhances termination
properties. Column ``TCLP'' in Table~\ref{tab:distance-summary}
summarizes the termination characteristics of \code{dist/3} under TCLP, and
shows that a full integration of tabling and CLP makes it possible to
find all the solutions and finitely terminate in all the
cases. Additionally, in~\cite{TCLP-tplp2019} we experimentally show
that Mod~TCLP, a framework that fully implements entailment in
the call and answer entailment phase,  can improve performance.

The theoretical basis of Tabled Constraint Logic Programming (TCLP)
were established in~\cite{toman_theo_const_tabling} using a framework
of bottom-up evaluation of Datalog systems and presenting the basic
operations (projection and entailment checking) that are necessary to
ensure completeness w.r.t.\ the declarative semantics.
In this work, we present the theoretical basis of TCLP for a top-down
execution on which Mod~TCLP~\cite{TCLP-tplp2019} is based.
In Section~\ref{sec:foundations} we present the operational semantics
of a top-down execution of TCLP programs with generic constraint
solvers.
In Section~\ref{sec:theorems-proofs} we extend the soundness,
completeness, and termination proofs.
In Section~\ref{sec:import-prec-impl} we explain the benefits of using
entailment checking with more relaxed notions projections.

\section{Fixpoint and Top-Down Semantics of TCLP}
\label{sec:foundations}

In this section we present a bottom-up fixpoint semantics of TCLP that
used constraint entailment for the answers and a top-down semantics
that extends~\cite{toman_theo_const_tabling} by explicitly modeling
entailment both in the answers and in the calls. This semantics uses
objects that mimic the construction of forests of trees in
implementations of tabling.

\subsection{Syntax of TCLP Programs}
\label{sec:tclp_syntax}

A (tabled) constraint logic program consists of clauses of the form:
$$ h\ \text{:-}\ c,\, l_1,\, \ \dots,\,\ l_k. $$
where
$h$ is an atom, $c$ is an atomic constraint or conjunction of constraints, 
$l_i$
are literals, `:-' represents the logical  implication `$\leftarrow$',
and `,' represents the logical conjunction
`$\land$'.
The head of the clause is $h$ and the rest is called
the body, denoted by $body(h)$.
We
will assume throughout this paper that the program has been rewritten
so that clause heads are linearized (all the
variables are different) and all head unifications
take place in $c$.  The constraint $c$ or the literals $l_i$ or both
may be absent. In the last case the rule is called a fact and it is
customarily written omitting the body.
We will assume that we are dealing with \emph{definite programs},
i.e., programs where the literals in the body are always positive
(non-negated) atoms.

A query to a TCLP program is a clause with the head \emph{false}, usually written
$\text{?-}\ c_q,\ q$, 
where $c_q$ is an atomic constraint or a conjunction of constraints and $q$ 
is a literal.\footnote{This covers as well the case of a conjunction of
  literals since we can always add a rule to that effect to the
  program.}
%
%



\subsection{Constraint Solvers}
\label{sec:const_solver}

We follow~\cite{survey94} in this section.
Constraint logic programming introduces constraint solving methods in
logic-based programming languages.
During the evaluation of a CLP program, the \emph{inference engine}
generates constraints whose consistency with respect to the current
constraint store are checked by the \emph{constraint solver}.  If the
check fails, the engine backtracks to a previous choice and takes a
pending, unexplored branch of the search tree.  In the next sections we will
review the fixpoint and operational semantics of CLP and will extend
them to TCLP.


\begin{deff}
\label{def:constraint-solver}
  A \emph{constraint solver}, CLP($\pazocal{X}$), is a (partial)
  executable implementation of a \emph{constraint domain}
  $(\domain, \pazocal{L})$. The parameter $\pazocal{X}$ stands for
  the 4-tuple ($\Sigma$, $\domain$, $\pazocal{L}$, $\pazocal{T}$)
  where:

  \begin{itemize}[wide=0.5em, leftmargin =*, nosep, before =
    \leavevmode\vspace{-\baselineskip}]
  \item[--] $\Sigma$ is a signature which determines the predefined
    predicates and function symbols and their arities.
  \item[--] $\domain$ is a $\Sigma$-structure: the constraint domain over
    which the computation is performed.
  \item[--] $\pazocal{L}$ is the class of $\Sigma$-formulas: the class
    of constraints that can be expressed with $\Sigma$. It should be
    closed under variable renaming, conjunction, and existential
    quantification. 
\item[--] $\pazocal{T}$ is a first-order $\Sigma$-theory: an
  axiomatization of the properties of $\domain$, which determines
  what constraints hold and what constraints do not hold.
  $\domain$ and $\pazocal{T}$ should agree on satisfiability of
  constraints, and every unsatisfiability in $\domain$ has to be
  detected by $\pazocal{T}$, i.e., for every constraint
  $c \in \pazocal{L}$, $\domain \vDash c$ iff
  $\pazocal{T} \vDash c$.

\end{itemize}

\end{deff}

A constraint can be an atomic constraint or a conjunction of (simpler)
constraints.  We denote constraints with lower case letters, e.g.\
$c$, and sets of constraints with uppercase letters, e.g.\ $S$.

\begin{exaa}
  \leftskip=20pt\parindent=0pt%
  ~\\
  The Herbrand domain
  CLP(\h) 
  used in logic programming is the constraint domain over finite
  trees, where $\Sigma$ contains constants, function symbols, and the
  predicate \code{=/2}; $\domain$ is the set of finite trees, where
  each node is labeled by a constant (if it does not have children) or
  a function symbol of arity $n$ (if it has $n$ children).
  $\pazocal{L}$ is the set of constraints generated by the primitive
  constraints (i.e., equality) between trees (terms).  Typical
  constraints are \code{X=g(a)} and %
  \code{X=f(Z,Y)$\;\land\;$Z=a}.
\end{exaa}

\begin{deff}[Valuation]\label{def:valuation}
  Let $S = \{X_1, \dots, X_n\}$ be a set of variables.  A
  \emph{valuation} $v$ is a mapping from variables in $S$ to values in
  $\domain$. We write
  $v = \{X_1 \mapsto d_1, \dots, X_n \mapsto d_n\}$ to indicate that
  the value $d_i$ is assigned to variable $X_i$.  
\end{deff}

For convenience, and where it is not ambiguous, we will denote the
value $d_i$ assigned to a variable $X_i$ by the valuation $v$ as
$v(X_i)$ (e.g., $X_i \mapsto d_i \in v$).  Likewise, for a literal $l$
we will denote by $v(l)$ the literal obtained by substituting the
variables in $l$ for their associated values in the valuation $v$ (for
those variables that appear in $v$) and, for a constraint $c$, we define
similarly $v(c)$.

\begin{deff}[Solution of a constraint]
  \label{def:solution}
  Let $c$ be a constraint, $vars(c)$ the set of variables occurring in
  $c$, and $v$ a valuation over $vars(c)$ on the constraint domain
  $\domain$. Then $v$ is a \emph{solution} for the constraint $c$ if
  $v(c)$
  holds in the constraint domain.
\end{deff}

\begin{deff}[Projection]
  \label{def:projection}
  Let $c$ be a constraint, $S \subseteq vars(c)$ a set of
  variables occurring in $c$, and $T = vars(c) \, \backslash \, S$ the
  rest of the variables of $c$.  The projection of $c$ over $S$,
  denoted $Proj(c, S)$, is another constraint $c_s$ such that
  $c_s \equiv \exists T \cdot c$, i.e.:

  \begin{itemize}[wide=0.5em, leftmargin =*, nosep, before =
    \leavevmode\vspace{-\baselineskip}]
  \item [--] Any solution $v_s$ for $c_s$ can be extended to be a solution
    for $c$.
  \item [--] Any solution $v$ for $c$ can be restricted to the variables in
    $S$ and the restricted valuation is a solution for $c_s$.
  \end{itemize}
\end{deff}

The minimal set of operations that we expect a constraint solver to
support, in order to interface it successfully with a tabling
system~\cite{TCLP-tplp2019}, are:
\begin{itemize}[noitemsep]
  
\item Test for consistence or satisfiability. A constraint $c$ is
  consistent in the constraint domain $\domain$, denoted $\domain \vDash c$, if
  it has a solution in $\domain$.

\item Test for entailment ($\sqsubseteq_{\domain}$).\footnote{We may
    omit the subscript \domain if there is no ambiguity.}  We say that
  a constraint $c_0$ is entailed by another constraint $c_1$
  ($c_0 \sqsubseteq_{\domain} c_1$) if any solution of $c_0$ is also a
  solution of $c_1$.
  We extend the notion of constraint entailment to a set of constraints: a set of
  constraints $C_0$ is entailed (or covered) by another set of
  constraints $C_1$ (and we write it as  $C_0 \sqsubseteq_{\domain}
  C_1$) if
  $\forall c_{i}\in C_0 \, \exists c_{j}\in C_1 . c_{i}
  \sqsubseteq_{\domain} c_{j}$. %

\item An operation to compute the projection of a constraint $c$ onto
  a finite set of variables $S$.
  $Proj(S,c)$.

\end{itemize}

\subsection{Fixpoint Semantics}
\label{sec:semantics}

The canonical model of a Prolog program is the minimal Herbrand model.
Similarly, the fixpoint semantics of a CLP program $P$ over a
constraint domain $\domain$ is the least 
$\domain$-S-model, which we define next.  The presence of
variables in $\domain$-S-models makes it possible to use
entailment to discard subsumed constraints in the bottom-up
construction of the fixpoint.




We can define the least $\domain$-S-model of a program using the
S-semantics~\cite{falaschi-vars89,survey94} for languages with
constraints~\cite{gabbrielli1991modeling}.  It differs from the
standard model~\cite{EmKo76} essentially due to the presence of
variables in interpretations and models.

\begin{deff}[$\domain$-S-interpretation]
  Let the pair $(l,c)$ be a constraint literal, where $l$ is a literal
  and $c \in \domain$ an atomic constraint or a conjunction of
  constraints such that $vars(c) \subseteq vars(l)$. A
  $\domain$-S-interpretation is a set of constraint literals.
\end{deff}

\begin{deff}[$\domain$-S-model]
  \label{def:s-model}
  Let $P$ be a program.  A $\domain$-S-model of $P$ is a
  $\domain$-S-interpretation that is logically consistent with the
  clauses in $P$.
\end{deff}

The CLP fixpoint S-semantics is defined as the smallest
fixpoint of the immediate consequence operator, $S_P^{\domain}$, where
all the operations behave as defined in the constraint domain
$\domain$.




\begin{deff}[Operator $S_P^{\domain}$~\cite{falaschi-vars89,toman_theo_const_tabling}]
\label{def:S_p}
  Let $P$ be a CLP program and $I$ a
  $\domain$-S-interpretation. The immediate consequence operator
  $S_P^{\domain}$ is defined as:

  \vspace*{0.5em}
  \begin{tabular}{ll}
    $S_P^{\domain}(I) = I \ \cup \ \{\ (h,c)\ \mid$
    & $h$ :- $c_h$, $l_1$, $\dots$, $l_k$ is a clause of $P,$ \\
    & $(a_i, c_i) \in I,\ 0 < i \leq k,$ \\
    & $ c' = Proj(vars(h),\ c_h\ \land\ \bigwedge_{i=1}^k (a_i = l_i\ \land\ c_i)),$\\
    & $\domain \vDash c',$ \\
    & $ \text{if } c' \sqsubseteq c'' \text{ for some } (h,c'') \in I
      \text{ then } c = c'' \text{ else } c = c'\ \} $\\
  \end{tabular}
\end{deff}

Note that
$S_P^\domain$ may not
add a pair \emph{(literal, constraint)} when a constraint more general
is already present in the interpretation being enlarged.  However, to
guarantee monotonicity, it does not remove existing more particular
constraints.  The operational semantics of TCLP
(Definition~\ref{def:tclp-forest}) will do that. %

\subsection{Operational Semantics of TCLP}
\label{sec:operational}

In this section we first present a top-down semantics for CLP without
tabling/suspension~\cite{survey94} and then we extend it to capture
the operational semantics
of TCLP.  The operational semantics is given in terms of a transition
system that computes the least model defined by the CLP fixpoint
semantics (Section~\ref{sec:semantics}). The evaluation of a query is
a sequence of steps from the initial state to a final state.

\begin{deff}
  \label{def:state}
  A \emph{state} is a tuple \goal{R,c} where:
  
  \begin{itemize}[wide=0.5em, leftmargin =*, nosep, before =
    \leavevmode\vspace{-\baselineskip}]
  \item[--] $R$, the resolvent, is a multiset of literals and
    constraints that contains the collection of as-yet-unseen literals and
    constraints of the program.

  \item[--] $c$, the constraint store, is an atomic constraint or a
    conjunction of constraints. It is acted upon by the solver.
  \end{itemize}
\end{deff}

In~\cite{survey94}
the constraint store is divided into a
collection of \emph{awake} constraints and a collection of
\emph{asleep} constraints.  This separation is ultimately motivated by
implementation issues and we will not make this distinction here.

Given a query \calltwo{q}{c_q},
the initial state of the evaluation is \goaltwo{q}{c_q}.
%
%
Every transition step between states resolves literals of the
resolvent against the clauses of the program
and adds constraints to the constraint store.
A derivation is \emph{successful} if it is finite and the final
state has the form \goalset{\emptyset}{c} (i.e., the resolvent becomes
empty). The answer for the query
is $Proj(vars(q), c)$.



As it is customary, we assume that the transitions due to constraint
handling are deterministic (there is only one possible children per
node), while the transitions due to literal matching may be
non-deterministic (there are as many children as clauses whose head
matches some literal in the resolvent).  
As a result, query evaluation takes the shape of 
a search tree, constructed following Def.~\ref{def:clp-tree}.
The order in which literals are selected is not relevant.  In
practice, implementations would use a \emph{computation rule} that is
in charge of deciding the new constraint/literal to be resolved among
the set of pending literals.  A common rule is to follow the
left-to-right order in which literals are written in the body of
clauses.



In what follows we will assume that variables in clauses are renamed
apart before they are used in order to avoid clashes with existing
variable names.

\begin{deff}[CLP tree]
  \label{def:clp-tree}
  Let $P$ be a CLP definite program and \calltwo{q}{c_q} a query. A
  CLP tree of \calltwo{q}{c_q} for $P$, denoted by
  $\bigsymbol{\tau}_P(q,c_q)$, is a tree such that:

  \begin{enumerate}[wide=0.5em, leftmargin =*, nosep, before =
    \leavevmode\vspace{-\baselineskip}]
  \item
    \label{item:query}
    The root of $\bigsymbol{\tau}_P(q,c_q)$ is
 \goaltwo{q}{c_q}, the initial state.
  \item
    \label{item:states}
    The nodes of $\bigsymbol{\tau}_P(q, c_q)$ are labeled with its
    corresponding state \goalset{L}{c}, where $L$ is a set containing the
    constraints and literals pending to be solved.
  \item
    \label{item:clauses}
    The child/children of a node \goalset{l \cup L}{c},
    where $l$ is a literal,
is/are:
    \begin{itemize}
    \item A node/nodes
 \goalset{body(h_i) \cup L}{c \land (l = h_i)} obtained
      by resolution of $l$ against the matching clause(s) $h_i$ :-
      $body(h_i)$ in $P$ where $l = h_i$ is an abbreviation for the
      conjunction of equations between the arguments of $l$ and
      $h_i$. There is one node for each matching clause.  Matching
      clauses are assumed to be renamed apart.
    \item Or a leaf node \emph{fail} if there are no clauses in $P$
      which matching heads for the literal $l$.
    \end{itemize}
  \item
    \label{item:const}
    The child of a node \goalset{c' \cup L}{c}, where $c'$ is a
    constraint,
    is:
    \begin{itemize}
    \item The node \goalset{L}{c \land c'} if
      $\domain \vDash c \land c'$.
    \item Or a leaf node \emph{fail} if
      $\domain \not\vDash c \land c'$.
    \end{itemize}
  \item
    \label{item:final}
    A leaf node \goalset{\emptyset}{c} is the final state of a
    successful derivation. $c$ is the final constraint store.
  \item
\label{item:ans}
    The set of answers of $\bigsymbol{\tau}_P(q, c_q)$ (i.e., the
    answers to the query \calltwo{q}{c_q}), denoted by $Ans(q, c_q)$, is the
    set of constraints $c'_{i}$ obtained as the projection of the
    final constraint stores $c_i$ onto $vars(q)$:
    $$Ans(q,c_q) = \{ c'_{i}\ \mid\ \ c'_{i} = Proj(vars(q),
    c_i).\goalset{\emptyset}{c_i} \in \bigsymbol{\tau}_P(q,c_q)
    \}$$
  \end{enumerate}
  
\end{deff}

We denote the set of tabled predicates in a TCLP program by
$Tab_P$.
The most general calls to predicates in $Tab_P$ are called generators
and are resolved against program clauses.
The set of generators created during the evaluation of a query
\calltwo{q}{c_q} is denoted by $Gen{(q, c_q)}$. 
The answers for a generator are collected and associated to that
generator; see below how entailment is used to keep only the relevant answers.
%
Calls to tabled predicates that are more particular than a previously
created generator become consumers and are not resolved against
program clauses. Instead, they are resolved by consuming the answers
collected from a generator; this is termed \emph{answer resolution}.

The execution of a query w.r.t.\ a TCLP program is represented as a
\emph{forest} of \emph{derivation trees}, and contains the tree
corresponding to the initial query and the trees corresponding to each
of the
generators.  The evaluation of each generator corresponds to one of
the trees of the forest.  During execution, call entailment
(Def.~\ref{def:tclp-forest}.\ref{item:call_entail})
detects when a goal is entailed/subsumed by a previous goal (its
generator) and if so, it suspends their execution and eventually
reuses the answers from the generator.  During answer entailment,
answers that are entailed by another (more general) answer are
discarded/removed (Def.~\ref{def:tclp-forest}.\ref{item:ans_entail}).



\begin{deff}[TCLP forest] 
  \label{def:tclp-forest}
  Let $P$ be a TCLP definite program, $Tab_P$ the set of tabled
  predicates, and \calltwo{q}{c_q} a query. A TCLP forest of \calltwo{q}{c_q} for
  $P$, denoted as $\mf_P(q,c_q)$ is the set of TCLP trees such that:
  
  \begin{enumerate}[wide=0.5em, leftmargin =*, nosep, before =
    \leavevmode\vspace{-\baselineskip}]

  \item The initial tree, $\bigsymbol{\tau}_P(q,c_q)$, is the TCLP
    tree of the query, and the rest of the trees,
    $\bigsymbol{\tau}_P(g_i, c_{g_i})$, are the TCLP trees of the
    generators \calltwo{g_i}{c_{g_i}} $\in Gen{(q, c_q)}$:
    \begin{displaymath}
      \mf_P(q,c_q) = \{ \bigsymbol{\tau}_P(q,c_q),
      \  \bigsymbol{\tau}_P(g_i, c_{g_i}), \dots \} \mbox{ with } i \geq 0
    \end{displaymath}

  \item A TCLP tree, denoted by $\bigsymbol{\tau}_P(q,c_q)$ (resp.\
    $\bigsymbol{\tau}_P(g_i, c_{g_i})$), is similar to a CLP tree
    where:
    
    \begin{enumerate}[wide=0.5em, leftmargin =*, nosep, before =
      \leavevmode\vspace{-\baselineskip}]
    \item
      \label{item:t-query}
      The root of the TCLP tree $\bigsymbol{\tau}_P(g,c)$ is
      \goaltwo{g}{c}, its initial state.


      
    \item
      \label{item:call_entail}
      The descendants of a node \goalset{t \cup L}{c} where $t$ is a
      tabled literal
      are obtained by obtaining answers for $t$ through answer
      resolution (i.e., consuming existing answers) in one of the two
      following ways:


      \begin{itemize}
      \item[--] If \calltwo{t}{c} is a consumer of a previous
        generator $\calltwo{g}{c_g} \in Gen{(q,c_q)}$, we use the
        answers $c_i \in Ans(g, c_g)$ to construct its children.
        In this case, $g$ and $t$ match and \calltwo{g}{c_g} is
        entailed by \calltwo{t}{c}, i.e.,
        \mbox{$c \land (t = g) \sqsubseteq_{\domain} c_g$}.
        %
        As a reminder, $t = g$ denotes the conjunction of equality
        constraints between the corresponding arguments of $t$ and $g$
        and $Ans(g, c_g)$ is the set of recorded answers for
        \calltwo{g}{c_g}.

      \item[--] Otherwise, \calltwo{t}{c} will produce a new generator
        \calltwo{t}{c'} and we use the answers $c_i \in Ans(t,
        c')$. In this case, a new TCLP tree
        $\bigsymbol{\tau}_P(t, c')$, %
        where $c' = Proj(vars(t),c)$, is created and added to the
        current forest.
        The goal \calltwo{t}{c'} is then marked as a generator and
        added to
        $Gen{(q,c_q)}$. 
      \end{itemize}

      From the possible answers $c_i$ to \calltwo{t}{c}, children nodes are
      constructed as follows:
      \begin{itemize}
      \item A node \goalset{c_i \cup L}{c}, one for each answer
        $c_i$.
      \item Or a leaf \emph{fail} if there is no answer $c_i$.
      \end{itemize}

    \item
      \label{item:rest1}
      The transitions for non-tabled literals and for new generators are
      as in the CLP tree (Def.~\ref{def:clp-tree}.\ref{item:clauses}).

    \item
      \label{item:rest2}
      The transitions for constraints are as in the CLP tree
      (Def.~\ref{def:clp-tree}.\ref{item:const}).

    \item
      \label{item:t-final}
      A leaf node \goalset{\emptyset}{c} is the final
      state of a successful derivation and $c$ is its final
      constraint store.
    
    \item 
      \label{item:ans_entail}
      The set of answers of $\bigsymbol{\tau}_P(g,c_g)$, the TCLP tree
      of the generator \calltwo{g}{c_g}, denoted by $Ans(g,c_g)$, is the set
      constraints $c'_i$ obtained as the projection of the final
      constraint stores $c_i$ onto $vars(g)$ that do not entail any
      other constraint $c_j$, i.e., they are the most general answers.

      \vspace{-2em}
      
      \begin{align*}
        Ans(c, c_g) = \{ c'_i & \mid c'_i = Proj(vars(g),c_i), 
        \goalset{\emptyset}{c_i} \in \bigsymbol{\tau}_P(g,c_g), \\
        & \nexists c_j \cdot \goalset{\emptyset}{c_j} \in \bigsymbol{\tau}_P(g,c_g), 
        c_i \not= c_j, c'_i \sqsubseteq Proj(vars(g),c_j)\}
      \end{align*}


      \end{enumerate}
    

    \item The set of the answers of the forest $\mf_P(q,c_q)$, denoted
      by $Ans(q,c_q)$, is the set of answers of
      $\bigsymbol{\tau}_P(q,c_q)$ that are obtained as in the CLP tree
      (Def.~\ref{def:clp-tree}.\ref{item:ans}).


  \end{enumerate}
  
\end{deff}

The answer management strategy used in 
Def.~\ref{def:tclp-forest}.\ref{item:ans_entail} aims
at keeping only the most general answers.
Since implementations incrementally save answers as they are found,
some previous proposals used simpler answer management strategies.
For example,~ \cite{bao2000,chico-tclp-flops2012-large} checked
entailment when adding answers to the previously generated ones and
only discarded answers which were more particular than a previous one.
This reduces the number of saved answers, but older answers that are
more particular than newer answers were still kept.  It could also be
possible to remove previous answers that are more particular than new
answers but still add answers that are more particular than previous ones.
The choice among them does not impact soundness or completeness
properties.  However, discarding \textbf{and} removing redundant
answers, despite extra cost, has been shown to
greatly increase the efficiency of the implementation~\cite{TCLP-tplp2019}.

\begin{figure}
  \begin{tabular}{p{.97\textwidth}}
    \\[.1cm]
    {

  \renewcommand{\newstate}[3]{ \textbf{#1} }
  \renewcommand{\newstatelarge}[3]{ \newstate{#1}{}{} }
  
  \begin{minipage}{.45\textwidth}

    \begin{forest}
      for tree={
        font=\sffamily,
        delay={
          edge label/.wrap value={node[midway,font=\sffamily, left]{#1}},
        },
        text width=2cm,
        text centered,
        edge={->,>=stealth}
      }
      [ \newstate{s1}%
      {dist(a,Y,D)}%
      {D<150}
      [ \newstate{s2i}%
      {D1\pr1\#>0, D2\pr1\#>0, D\pr1\#=D1\pr1+D2\pr1, edge(a,Z\pr1,D1\pr1), dist(Z\pr1,Y\pr1,D2\pr1)}%
      {D<150,Y\pr1=Y,D\pr1=D}
      [ \newstate{s3}%
      {edge(a,Z\pr1,D1\pr1), dist(Z\pr1,Y,D2\pr1)}%
      {D<150,D1\pr1>0,D2\pr1>0,D=D1\pr1+D2\pr1}
      [ \newstate{s4}%
      {dist(b,Y,D2\pr1)}%
      {D<150,D1\pr1>0,D2\pr1>0,D=D1\pr1+D2\pr1,Z\pr1=b,D1\pr1=50}
      [ {Ans({\bf s5})}, delay={text width=3cm, text centered, edge={-}} 
      [ \newstate{s11}%
      {\ }%
      {D<150,D2\pr1>0,D=50+D2\pr1,V0=a,D2\pr1>25,D2\pr1<35}, edge label=(b1)
      [ {\textbf{a2} $\thickmuskip=0mu \mt{V0=a\ \land}$\\ $\thickmuskip=0mu \mt{V1>75 \land V1<85}$}, delay={text width=2cm, text centered}  ]
      ]
      [ \newstate{s14}%
      {\ }%
      {D<150,D2\pr1>0,D=50+D2\pr1,V0=b,D2\pr1>75,D2\pr1<85}, edge label=(b2)
      [ {\textbf{a3} $\thickmuskip=0mu \mt{V0=b\ \land}$\\ $\thickmuskip=0mu \mt{V1>125 \land  V1<135}$}, delay={text width=2.2cm, text centered}  ]
      ]
      ]
      ]
      ]
      ]
      [ \newstate{s2ii}%
      {edge(a,V0,D)}%
      {D<150,Y\pr1=V0,D\pr1=D}
      [ \newstate{s10}%
      {\ }%
      {D<150,V0=b,D=50}
      [  {\textbf{a1} $\thickmuskip=0mu \mt{V0=b\ \land}$\\ $\thickmuskip=0mu \mt{V1=50}$}, delay={text width=2cm, text centered}  ]
      ]
      ]
      ]
    \end{forest}

  \end{minipage}
  \begin{minipage}{.45\textwidth}

    \begin{forest}
      for tree={
        font=\sffamily,
        delay={
          edge label/.wrap value={node[midway,font=\sffamily, left]{#1}},
        },
        edge={->,>=stealth}
      }
      [ \newstate{s5}%
      {dist(b,Y,D)}%
      {D>0, D<100}
      [ \newstate{s6i}%
      {D1\pr1\#>0, D2\pr1\#>0, D\pr1\#=D1\pr1+D2\pr1, edge(b,Z\pr1,D1\pr1), dist(Z\pr1,Y\pr1,D2\pr1)}%
      {D>0, D<100,Y\pr1=Y,D\pr1=D}
      [ \newstate{s7}%
      {edge(b,Z\pr1,D1\pr1), dist(Z\pr1,Y,D2\pr1)}%
      {D>0, D<100,D1\pr1>0,D2\pr1>0,D=D1\pr1+D2\pr1}
      [ \newstate{s8}%
      {dist(a,Y,D2\pr1)}%
      {D>0, D<100,D1\pr1>0,D2\pr1>0,D=D1\pr1+D2\pr1,Z\pr1=a,D1\pr1>25,D1\pr1<35}
      [ {Ans({\bf s1})}, delay={text width=2cm, text centered, edge={-}} 
      [ \newstate{s12}%
      {\ }%
      {D>0, D<100,D2\pr1>0,D>25+D2\pr1,D<35+D2\pr1,Y=b,D2=50}, edge label=(a1)
      [  {\textbf{b2} $\thickmuskip=0mu \mt{V2=b\ \land}$\\ $\thickmuskip=0mu \mt{V3>75 \land  V3<85}$}, delay={text width=2cm, text centered}  ]
      ]
      [ \newstate{s13}%
      {\ }%
      {D>0, D<100,D2\pr1>0,D>25+D2\pr1,D<35+D2\pr1,Y=a,D2>75,D2<85}, edge label=(a2)
      [  {\bf fail} ]
      ]
      [ \newstate{s15}%
      {\ }%
      {D>0, D<100,D2\pr1>0,D>25+D2\pr1,D<35+D2\pr1,Y=b,D2>125,D2<135}, edge label=(a3), delay={xshift=1cm}
      [  {\bf fail}, delay={xshift=1cm} ]
      ]
      ]
      ]]]
      [ \newstate{s6ii}%
      { edge(b,Y,D)}%
      { D>0, D<100, Y\pr1=Y, D\pr1=D }
      [ \newstate{s9}%
      {\ }%
      { D>0, D<100, Y=a, D>25, D<35 }
      [ {\textbf{b1} $\thickmuskip=0mu \mt{V2=a\ \land}$\\ $\thickmuskip=0mu \mt{V3>25 \land  V3<35}$}, delay={text width=3cm, text centered}  ]
      ]
      ]]
    \end{forest}

  \end{minipage}
} \\
    \\[-.2em]
    {\vspace{0.5em}\begin{forest}
  for tree={
    font=\sffamily,
    delay={
      edge label/.wrap value={node[midway,font=\sffamily, left]{#1}},
    },
    folder,
    grow'=0,
    fit=band,
    l sep'-=1pt,
    s sep+=-1pt,
    inner sep=1pt,
    text width=20cm,
    text height=.1cm,
    edge={->,>=stealth}
  }
  [ \newstate{s1}%
  {dist(a,V0,V1)}%
  {V1<150}
  [ \newstatelarge{s2i}%
  {D1\pr1\#>0, D2\pr1\#>0, V1\pr1\#=D1\pr1+D2\pr1, edge(a,Z\pr1,D1\pr1), dist(Z\pr1,Y\pr1,D2\pr1)}%
  {V1<150 \land Y\pr1=V0 \land D\pr1=V1}
  [ \newstate{s3}%
  {edge(a,Z\pr1,D1\pr1), dist(Z\pr1,V0,D2\pr1)}%
  {V1<150 \land D1\pr1>0 \land D2\pr1>0 \land V1=D1\pr1+D2\pr1}
  [ \newstate{s4}%
  {dist(b,V0,D2\pr1)}%
  {V1<150 \land D1\pr1>0 \land D2\pr1>0 \land V1=D1\pr1+D2\pr1 \land Z\pr1=b \land D1\pr1=50}, delay={l sep+=5pt,s sep+=7pt}
  [ {Ans(dist(b,V2,V3), $V3>0 \land  V3<100$)},%
  delay={%
    edge={-},%
    edge label={node[midway,font=\sffamily, right, yshift=10pt]%
      {\smaller with renaming $V0=V2\ \land\  D2\pr1=V3$}}}
  [ \newstate{s11}%
  {\,}%
  {V1<150 \land D2\pr1>0 \land V1=50+D2\pr1 \land V0=a \land D2\pr1>25 \land D2\pr1<35}, edge label=(b1)
  [ {\textbf{a2} $\thickmuskip=0mu \mt{V0=a \land  V1>75 \land  V1<85}$} ]
  ]
  [ \newstate{s14}%
  {\,}%
  {V1<150 \land D2\pr1>0 \land V1=50+D2\pr1 \land V0=b \land D2\pr1>75 \land D2\pr1<85}, edge label=(b2)
  [ {\textbf{a3} $\thickmuskip=0mu \mt{V0=b \land  V1>125 \land  V1<135}$} ]
  ]
  ]
  ]
  ]
  ]
  [ \newstate{s2ii}%
  {edge(a,V0,V1)}%
  {V1<150 \land Y\pr1=V0 \land D\pr1=V1}
  [ \newstate{s10}%
  {\,}%
  {V1<150 \land V0=b \land V1=50}
  [  {\textbf{a1} $\thickmuskip=0mu \mt{V0=b \land  V1=50}$} ]
  ]
  ]
  ]
\end{forest}} \\[-.3em]
    {\vspace{0.5em}\begin{forest}
  for tree={
    font=\sffamily,
    delay={
      edge label/.wrap value={node[midway,font=\sffamily, left]{#1}},
    },
    folder,
    grow'=0,
    fit=band,
    l sep'-=1pt,
    s sep+=-1pt,
    inner sep=1pt,
    text width=20cm,
    text height=.1cm,
    edge={->,>=stealth}
  }
  [ \newstate{s5}%
  {dist(b,V2,V3)}%
  {V3>0 \land  V3<100}
  [ \newstatelarge{s6i}%
  {D1\pr1\#>0, D2\pr1\#>0, D\pr1\#=D1\pr1+D2\pr1, edge(b,Z\pr1,D1\pr1), dist(Z\pr1,Y\pr1,D2\pr1)}%
  {V3>0 \land  V3<100 \land Y\pr1=V2 \land D\pr1=V3}
  [ \newstate{s7}%
  {edge(b,Z\pr1,D1\pr1), dist(Z\pr1,V2,D2\pr1)}%
  {V3>0 \land  V3<100 \land D1\pr1>0 \land D2\pr1>0 \land V3=D1\pr1+D2\pr1}
  [ \newstate{s8}%
  {dist(a,V2,D2\pr1)}%
  {V3>0 \land  V3<100 \land D1\pr1>0 \land D2\pr1>0 \land V3=D1\pr1+D2\pr1 \land Z\pr1=a \land D1\pr1>25 \land D1\pr1<35}, delay={l sep+=5pt,s sep+=7pt}
  [ {Ans(dist(a,V0\pr1,V1\pr1), $V1\pr1<150$) \smaller%
    is entailed  because $V1\pr1>0 \land  V1\pr1<75$ $\sqsubseteq$ $V1\pr1<150$},%
  delay={%
    edge={-},%
    edge label={node[midway,font=\sffamily, right, yshift=10pt]%
      {\smaller with renaming $V2=V0\pr1\ \land\  D2\pr1=V1\pr1$}}}
  [ \newstate{s12}%
  {\,}%
  {V3>0 \land  V3<100 \land D2\pr1>0 \land V3>25+D2\pr1 \land V3<35+D2\pr1 \land V2=b \land D2\pr1=50}, edge label=(a1)
  [  {\textbf{b2} $\thickmuskip=0mu \mt{V2=b \land  V3>75 \land  V3<85}$} ]
  ]
  [ \newstate{s13}%
  {\,}%
  {V3>0 \land  V3<100 \land D2\pr1>0 \land V3>25+D2\pr1 \land V3<35+D2\pr1 \land V2=a \land D2\pr1>75 \land D2\pr1<85}, edge label=(a2)
  [  {\bf fail} ]
  ]
  [ \newstate{s15}%
  {\,}%
  {V3>0 \land  V3<100 \land D2\pr1>0 \land V3>25+D2\pr1 \land V3<35+D2\pr1 \land V2=b \land D2\pr1>125 \land D2\pr1<135}, edge label=(a3)
  [  {\bf fail} ]
  ]
  ]
  ]]]
  [ \newstate{s6ii}%
  {edge(b,V2,V3)}%
  {V3>0 \land  V3<100 \land  Y\pr1=V2 \land  D\pr1=V3 }
  [ \newstate{s9}%
  {\,}%
  {V3>0 \land  V3<100 \land  V2=a \land  V3>25 \land  V3<35 }
  [ {\textbf{b1} $\thickmuskip=0mu \mt{V2=a \land  V3>25 \land  V3<35}$} ]
  ]
  ]]
\end{forest}} \\
  \end{tabular}
  \caption[TCLP forest of \code{dist/3} with right recursion.]{TCLP
    forest of \code{?-D.<.150,dist(a,Y,D)} with right recursion.}
  \label{fig:tree04}
\end{figure}

\begin{exaa}TCLP forest of \code{dist/3}
  \leftskip=20pt\parindent=0pt%
  ~\\
  This example 
  illustrates how the algorithm works with mutually dependent
  generators, i.e., generators that consume answers from each other,
  and to see why not all the answers from a generator may be directly
  used by its consumers.\footnote{This example also appears in the
    Supplementary Material of~\cite{TCLP-tplp2019}.}
  Fig.~\ref{fig:tree04} shows the TCLP forest corresponding to
  querying the right-recursive \code{dist/3} program
  (Fig.~\ref{fig:dist-clp-right}).
  Unlike the left-recursive version, which generates only one TCLP
  tree,
  the right-recursive version generates two TCLP trees, one for each
  generator.
  The reason is that the left-recursive version only seeks paths from
  the node \code{a}, but the right-recursive version creates a new
  TCLP tree at the state \textsf{s4} to collect the paths from the
  node \code{b}, since \code{edge(a,b)} had been previously evaluated
  at state \textsf{s3}. We explain now how we obtain some of the
  states; the rest are obtained similarly. 

\begin{description}[noitemsep, leftmargin=!,labelwidth=\widthof{\bfseries Ans(xx)
  }, labelindent=0cm, align=right]
  
\item[s1] the TCLP tree
  $\bigsymbol{\tau}_P(\mt{dist(a,V0,V1),\, V1<150})$ is created.

\item[s4] is obtained by resolving the literal
  \code{edge(a,Z$_1$,D1$_1$)}.

\item[Ans(s5)] the tabled literal \code{dist(b,V0,D2$_1$)} is a new
  generator and a new TCLP tree
  $\bigsymbol{\tau}_P(\mt{dist(b,V2,V3),\, V3>0 \land V3<100})$ is created
  (Def.~\ref{def:tclp-forest}.\ref{item:call_entail}).

\item[s5] is the root node of the new TCLP tree.
  
\item[s6i/ii] are obtained by resolving the literal
  \code{dist(b,V2,V3)} against the clauses of the program.
  
\item[s8] is obtained by resolving the literal
  \code{edge(b,Z$_1$,D1$_1$)}.

In the state \textsf{s8}, the call
\callcode{dist(a,V2,D2$_1$)}{D2$_1$>0$\ \land\ $D2$_1$<75} is suspended
because it entails the former generator
\callcode{dist(a,V0$_1$,V1$_1$)}{V1$_1$<150}.  
  
\item[Ans(s1)] the tabled literal \code{dist(a,V2,D2$_1$)} is
  resolved with answer resolution
  (Def.~\ref{def:tclp-forest}.\ref{item:ans_entail}) using the
  answers from the previous TCLP tree
  $\bigsymbol{\tau}_P(\mt{dist(a,V0_1,V1_1), V1_1<150})$ because the
  renamed projection\footnote{The projection of
    \code|V3>0$\ \land\ $V3<100$\ \land\ $D1$_1$>0$\ \land\ $D2$_1$>0$\ \land\ $V3=D1$_1$+D2$_1$$\ \land\ $Z$_1$=a$\ \land\ $D1$_1$>25| \code|$\land\ $D1$_1$<35| onto \code|D2$_1$| is \code|D2$_1$>0$\ \land\ $D2$_1$<75|.  After
    renaming \code|D2$_1$=V1$_1$|, the resulting projection is
    \code|V1$_1$>0$\ \land\ $V1$_1$<75|.} of the current constraint store onto
  the variable of the literal entails the projected constraint
  store of the generator:
  (\code{V1$_1$>0$\ \land\ $V1$_1$<75}) $\sqsubseteq$ \code{V1$_1$<150}. Since the initial
  TCLP forest is under construction and depends on itself, the current
  branch derivation is suspended.

  This suspension also causes the former generator to suspend at the
  state \textsf{s4}. 

\item[s9] is a final state obtained upon backtracking to the state
  \textsf{s6ii}.

\item[b1] is the first answer of the second generator.

  At this point the suspended calls can be resumed by consuming the
  answer \textsf{b1} or by evaluating \textsf{s2ii}. The algorithm
  first tries to evaluate \textsf{s2ii} and then it will resume
  \textsf{s4} consuming \textsf{b1}.

\item[s10] is a final state obtained upon backtracking to the state
  \textsf{s2ii}.
  
\item[a1] is the first answer of the first generator:
  \code{V0=b$\ \land\ $V1=50}.
  
\item[s11] is a final state obtained from the state \textsf{s4} by
  consuming \textsf{b1}.
  
\item[a2] is the second answer of the first generator:
  \code{V0=a$\ \land\ $V1>75$\ \land\ $V1<85}.
  
\item[s12] is a final state obtained from the state \textsf{s8} by
  consuming \textsf{a1}.
  
\item[b2] is the second answer of the second generator.

\item[s13] is a failed derivation obtained from \textsf{s8} by consuming
  \textsf{a2}. It fails because the constraints
  \code{V0=a$\ \land\ $V1>75$\ \land\ $V1<85} are inconsistent with the current
  constraint store. Note that the projection of the constraint store
  of \textsf{s8} onto \code{V1} is \code{V1>0$\ \land\ $V1<75}. Its child is
  a \code{fail} node.
  
\item[s14] is a final state obtained from the state \textsf{s4} by
  consuming \textsf{b2}.
  
\item[a3] is the third answer of the first generator:
  \code{V0=b$\ \land\ $V1>125$\ \land\ $V1<135}.

\item[s15] is a failed derivation obtained from \textsf{s8} by
  consuming \textsf{a3}. Its child is a \code{fail} node.

\end{description}

\end{exaa}

The comparison of this forest (with two trees) with the forest
obtained for the left-recursive version (with one tree) illustrates
why left recursion reduces the execution time and memory requirements
when using tabling / TCLP: left recursion will usually create fewer
generators.  We have also seen that using answers from a most general
call, as in the answer resolution of state \textsf{s8} (i.e., the
constraint store of the consumer \code{V1$_1$>0$\ \land\ $V1$_1$<75}
is more particular than the constraint store of the generator
\code{V1$_1$<150}), makes it necessary to filter the correct ones
(i.e., answer resolution for \textsf{a2} and \textsf{a3} failed). This
is not required in variant tabling because the answers from a
generator are always valid for its consumers.

\section{Soundness, Completeness, and Termination}
\label{sec:theorems-proofs}

In this section we prove the soundness and completeness
of the operational semantics for the top-down execution of tabled
constraint logic programs previously presented. Then, we present some
additional results on termination properties
for arbitrary constraint solvers that are not necessarily
constraint-compact, extending the results
in~\cite{toman_theo_const_tabling}.

\subsection{Soundness and Completeness}
\label{sec:soundn-compl}

\cite{toman_theo_const_tabling} proves soundness and completeness of
$SLG^C$ for TCLP Datalog programs by reduction to soundness and
completeness of bottom-up evaluation. It is possible to extend these
results to prove the soundness and completeness of our proposal: they only
differ in the answer management strategy and the construction of the
TCLP forest.  The strategy used in $SLG^C$ only discards answers which
are more particular than a previous answer, while in our proposal we
in addition remove previously existing more particular answers
(Def.~\ref{def:tclp-forest}.\ref{item:ans_entail}).  The result of
this is that only the most general answers are kept.
In $SLG^C$, the generation of the forest is modeled as the application
of rewriting rules.  In our proposal, the TCLP forest is defined as a
transition system (Def.~\ref{def:tclp-forest}), where the different
cases in the definition can be seen as rules which make the TCLP
forest evolve.

The lemma, theorems, and their proofs are reformulated taking in
consideration these differences.  First we prove that answer
resolution using entailment is correct w.r.t.\ SLD resolution; and
although only the most general answers are kept, answer resolution
using entailment is complete w.r.t.\ SLD resolution.  Then we use
these results to prove soundness and completeness of TCLP with
entailment w.r.t.\ the least fixed point semantics.

\begin{lemm}[Application of derivations with most general constraint stores]
  \label{lem:composition}
  Let \goaltwo{l_i,\,l_{i+1},\, \dots,\, l_k}{cs_{i}} $\leadsto$
  \goaltwo{l_{i+1},\, \dots,\, l_k}{cs_{i+1}} be a derivation
  and \calltwo{l_i}{c}
 a goal with  $cs_{i} \sqsubseteq c$. Then:
  %
  \begin{displaymath}
    \exists \goaltwo{l_i}{c} \leadsto \goalset{\emptyset}{c'} \mbox{ with } \ cs_{i+1}
    = cs_{i} \land c'
  \end{displaymath}
  \noindent
\end{lemm}


Intuitively, if there is an SLD derivation that gives a solution for a
goal \calltwo{l_i}{cs_i}, this solution can be obtained using the
solution for a more general goal \calltwo{l_i}{c} without the need to
resolve the more particular one.

\begin{proof}
We will see that there exists a  derivation \goaltwo{l_i}{c}
  $\leadsto$ \goalset{\emptyset}{c'} that follows the same steps as
  \goaltwo{l_i,\, \dots,\, l_k}{cs_{i}} $\leadsto$
  \goaltwo{l_{i+1},\, \dots,\, l_k}{cs_{i+1}}:
  
  (1) if \goaltwo{l_i,\, \dots,\, l_k}{cs_{i}} is
  resolved against a clause $l_i$ :- $c_h$, then its resulting
  constraint store is $cs_{i+1} = cs_{i} \land c_h$ (plus head unification).
  Since $cs_{i}
  \sqsubseteq c $, we can apply the same rule to 
  \goaltwo{l_i}{c} and its resulting constraint store is
  $c' = c \land c_h$. Also, since $cs_{i} \sqsubseteq c $, we have
  $cs_{i} \Leftrightarrow cs_{i} \land c$.  Therefore,
  $cs_{i+1}= cs_{i} \land c \land c_h$ (expanding $cs_{i}$) and
  $ cs_{i+1} = cs_{i} \land c'$ (contracting $c \land
  c_h$). 
  
  (2) if \goaltwo{l_i,\, \dots,\, l_k}{cs_{i}} is
  resolved against a clause $l_i$ :- $c_h,\, a_1,\, \dots,\, a_m$,
  the next state is \goaltwo{a_{1},\, \dots,\, a_m,\,
    l_{i+1},\, \dots ,\, l_k}{cs_{i} \land c_h} (resp.
  \goaltwo{a_{1},\, \dots,\, a_m}{c \land c_h}). By
  induction, since $cs_{i} \sqsubseteq true$ (resp.
  $c \sqsubseteq true$), there exist $m$ derivations
  $\goaltwo{a_j}{true} \leadsto \goalset{\emptyset}{c'_{a_{j}}}$ such that
  the resulting constraint store of the path is
  $cs_{i+1} = cs_{i} \land c_h \land \bigwedge_{j=1}^m c'_{a_j}$ (resp.
  $c' = c \land c_h \land \bigwedge_{j=1}^m c'_{a_j}$). Since
  $cs_{i} \sqsubseteq c $, we have
  $cs_{i} \Leftrightarrow cs_{i} \land c$. Therefore,
  $cs_{i+1} = cs_{i} \land c \land c_h \land \bigwedge_{j=1}^m
  c'_{a_j}$ (expanding $cs_{i}$) and $ cs_{i+1} = cs_{i} \land c'$
  (contracting $c \land c_h \land \bigwedge_{j=1}^m c'_{a_j}$).
\end{proof}

We will use this lemma to prove correctness of answer resolution.  We
model the answers obtained for a generator with the derivation
$\goaltwo{l_i}{c} \leadsto \goalset{\emptyset}{c'}$, while
$\calltwo{l_i}{cs_i}$ would be a consumer for the generator
\calltwo{l_i}{c}.  Note that the condition $cs_{i} \sqsubseteq c$
precisely captures the generator / consumer relationship.

\begin{corr}[Correctness of answer resolution using entailment]
  \label{cor:ans_corr}
  As an immediate consequence of Lemma~\ref{lem:composition}, using
  answer resolution with entailment
  (Def.~\ref{def:tclp-forest}.\ref{item:call_entail}) gives correct
  results.  Answer resolution of \goaltwo{l_i,\, \dots,\,
    l_k}{cs_{i}} consumes an answer $c'$ from a previous
  derivation \goaltwo{l_i}{c} $\leadsto$ \goalset{\emptyset}{c'} where
  \calltwo{l_i}{c} is the generator of the derivation and, by the definition
  of generator,
  $cs_{i} \sqsubseteq c$.
When
  $\domain \vDash cs_i \land c'$
  (Def.~\ref{def:tclp-forest}.\ref{item:rest2}), it generates the state
  \goaltwo{l_{i+1},\, \dots,\, l_k}{cs_{i} \land c'}.
\end{corr}

\begin{corr}[Completeness of answer resolution using entailment]
  \label{cor:ans_compl}
  Recall that  
$Ans(l, c)$ is the set containing the most general answers for a
generator goal \calltwo{l}{c}
(Def.~\ref{def:tclp-forest}.\ref{item:ans_entail}), and if there are
two goals \calltwo{l}{c_a} and \calltwo{l}{c_b} with $c_a \sqsubseteq c_b$,
only the answers for the most general goal $c_b$ need to be kept.  Therefore, for
any derivation of a generator $\goaltwo{l_i}{c} \leadsto \goalset{\emptyset}{c_i}$ we have
that $\exists c_i' \in Ans(l_i,c').c_i \sqsubseteq c_i'$ for some $c'$
s.t.\ $c \sqsubseteq c'$.  Let us take
a (partial) clause derivation
$\goaltwo{l_i, \ldots, l_k}{c} \leadsto \goaltwo{l_{i+1}, \ldots, l_k}{c
  \land c_i}$.  If $c_i' \in Ans(l_i, c')$ for some $c'$
s.t. \ $c \sqsubseteq c'$ (which is the entailment condition necessary
to use the saved answer constraints), then $c_i \sqsubseteq c_i'$.  If
we use $c_i'$ to perform answer resolution with \calltwo{l_i}{c}, we have
$\goaltwo{l_i, \ldots, l_k}{c} \leadsto \goaltwo{l_{i+1}, \ldots, l_k}{c
  \land c_i'}$.  Given that $c_i \sqsubseteq c_i'$, we have that
$c \land c_i \sqsubseteq c \land c_i'$, and any answer returned by
clause resolution is contained in some answer returned by answer
resolution with entailment.  The same reasoning can be applied to the
derivation of $l_{i+1}$ and so on.  Therefore, answer resolution with
entailment does not lose answers w.r.t.\ clause resolution even if not
all the goals and answers are memorized.
\end{corr}

\begin{thee}[Soundness w.r.t.\ the fixpoint semantics]
  \label{the:sound}
  Let $P$ be a TCLP definite program and \calltwo{q}{c_q} a query. Then for
  any answer $c'$ of the TCLP forest $\mf_P(q,c_q)$
  \begin{displaymath}
    c' \in Ans(q,c_q)\ \Rightarrow \exists
    \calltwo{q}{c} \in \lfp(S_P^{\domain}(\emptyset)).
    \ c' = c_q \land c
  \end{displaymath}
\noindent 
I.e., any answer derived from the forest construction can also be
derived from the bottom-up computation.
\end{thee}

\begin{proof}
  For any answer $c' \in Ans(q,c_q)$ there exists a successful
  derivation $\goaltwo{q}{c_q} \leadsto \goalset{\emptyset}{c'}$.  Since
  $c_q \sqsubseteq true$, by Lemma~\ref{lem:composition} there exists
  $\goaltwo{q}{true} \leadsto \goalset{\emptyset}{c}. \ c' = c_q \land
  c$. We know that for any successful derivation
  $\goaltwo{q}{true} \leadsto \goalset{\emptyset}{c}$ against the clauses
  of the program there is an answer derived
  from the bottom-up computation
  $\calltwo{q}{c} \in \lfp(S_P^{\domain}(\emptyset))$. Therefore, by
  Corollary~\ref{cor:ans_corr} if answer resolution is used instead of
  clause resolution, the result is also correct and for any answer
  $c' \in Ans(q,c_q)$ there exists
  $\calltwo{q}{c} \in \lfp(S_P^{\domain}(\emptyset)).  \ c' = c_q \land c$.
\end{proof}

\begin{thee} [Completeness w.r.t.\ the fixpoint semantics]
  \label{the:complete}
  Let $P$ be a TCLP definite program and \calltwo{h}{true} a query. Then for
  every \calltwo{h}{c} in $\lfp(S_P^{\domain})$:
  \begin{displaymath}
    \calltwo{h}{c} \in \lfp(S_P^{\domain}(\emptyset))\ \Rightarrow \ \exists c' \in Ans(h,true). \ c \sqsubseteq c'
  \end{displaymath}
  \noindent
  I.e., all the answers derived from the bottom-up computation are
  also derived by the forest construction or entailed by answers
  inferred in the forest.
\end{thee}

\begin{proof}
  We know that for any answer derived from the bottom-up computation
  $\calltwo{h}{c} \in \lfp(S_P^{\domain}(\emptyset))$ there exists a successful
  derivation $\goaltwo{h}{true} \leadsto \goalset{\emptyset}{c}$ against the
  clauses of the program.  By Corollary~\ref{cor:ans_compl} if answer
  resolution is used instead of clause resolution, the results is also
  complete. Therefore, since the answer management strategy only keeps
  the most general answers
  (Def.~\ref{def:tclp-forest}.\ref{item:ans_entail}), we have that
  $\exists c' \in Ans(h,true). \ c \sqsubseteq c'$.
\end{proof}

\subsection{Termination}
\label{sec:termination}

The next definition is a fundamental property of some constraint
domains that plays a key role in the termination of the evaluation of
queries to TCLP programs~\cite{toman_theo_const_tabling}.

\begin{deff}[Constraint-compact]
  \label{def:constraint-compact}
  Let \domain be a constraint domain, and $D$ the set of all
  constraints expressable in \domain. Then $\domain$ is
  constraint-compact iff:
  
  \begin{itemize}[wide=0.5em, leftmargin =*, nosep, before =
    \leavevmode\vspace{-\baselineskip}]
  \item[--] for every finite set of variables $S$, and
  \item[--] for every subset $C \subseteq D$ such that
    $\forall c \in C.vars(c) \subseteq S$,
  \end{itemize}
  
  there is a finite subset $C_{fin} \subseteq C$ such that
  $\forall c \in C.\exists c' \in C_{fin}. c \sqsubseteq_{\domain} c'$
\end{deff}

Intuitively speaking, a constraint domain \domain is
constraint-compact if for any (potentially infinite) set of
constraints $C$ expressable in \domain using a finite number of
variables, there is a \emph{finite} set of constraints
$C_{fin} \subseteq C$ that covers $C$ in the sense of
$\sqsubseteq_{\domain}$.  In other words, $C_{fin}$ is as general as
$C$.
Additionally,
in a constraint-compact constraint domain, if an infinite set of constraints is
unsatisfiable, then there is a finite subset which is unsatisfiable,
therefore guaranteeing the existence of finite unsatisfiability
proofs.

\begin{exaa}
  \leftskip=20pt\parindent=0pt%
  ~\\
  The gap-order constraints~\cite{revesz1993closed} is a
  constraint-compact domain generated from the set
  $\pazocal{C}_{<Z} = \{ x < u : u \in A\} \cup \{u < x : u \in A\}
  \cup \{x + k < y : k \in Z^+ \}$ where $A \subset Z^+$ is
  finite. First, we see that the set $C_{x<u}$ (resp. $C_{u<x}$) of
  possible constraints of the form $x < u$ (resp. $u < x$), where
  $x \in S$, is finite, because $A$ and $S$ are finite. Therefore, it
  is trivial to define a finite set that covers
  $C_{x<u} \cup C_{u<x}$.  Second, for every pair of variables
  $x, y \in S$, the set $C_{x+k<y}$ of possible constraints of the form
  $x + k < y, \, k \in Z^+$ can be covered by a finite subset of
  itself. Although for a given pair of variables $x$, $y$ one can
  generate an infinite number of constraints $x + k_i < y$ choosing
  different $k_i \in Z^+$, the constraint $x + k_0 < y$ having the
  smallest $k_0$ among all the $k_i$ ($\forall k_i. k_0 \leq k_i$)
  subsumes all the rest of the constraints
  ($x + k_i < y \ \sqsubseteq \ x + k_0 < y$).  Note that $k_0$ always
  exists, since $k_i \in Z^+$, which has a minimum.  Since $S$ is
  finite, we only have to check it for two given $x$, $y$; we can
  repeat the same process for every pair of variables, since there is
  only a finite number of them.  Therefore, the infinite set
  $C_{x+k<y}$ has a finite subset $C_{fin} = \{x + k_0 < y\}$ which
  covers it ($C_{x+k<y} \sqsubseteq C_{fin}$).
\end{exaa}

\begin{exaa}
  \leftskip=20pt\parindent=0pt%
  ~\\
  The Herbrand domain is not constraint-compact.  Take the
  infinite set of constraints
  $C = \{X = a,\, X = f(a),\, X = f(f(a)),\, \ldots\}$.
  No finite subset of $C$ using only constraints in $C$ can cover
  $C$.
\end{exaa}


The termination of TCLP Datalog programs under a top-down strategy
when the constraint system is constraint-compact is proven in
\cite{toman_theo_const_tabling}.
In that case, the evaluation will suspend the exploration of a call
whose constraint store is less general than or comparable to a previous
call.  Eventually, the program 
will generate a set of call constraint stores that can cover any
infinite set of constraints in the constraint domain, therefore finishing
evaluation.

Many TCLP applications require constraint domains that are not
constraint-compact because constraint-compact domains in general have
a limited expressiveness.  We refine here the termination
theorem~\cite[Theorem~23]{toman_theo_const_tabling} for Datalog
programs with constraint-compact domains to cover cases where the
constraint domain is not constraint-compact, but in which the program
evaluation generates only a constraint-compact subset of
all the constraints expressable in the constraint domain.

\begin{thee} [Termination in non constraint-compact domains]
  \label{th:termination}
  Let $P$ be a TCLP($\domain$) definite program and \calltwo{q}{c_q} a
  query. Then the TCLP execution for that query terminates iff:

  \begin{itemize}[wide=0.5em, leftmargin =*, nosep, before =
    \leavevmode\vspace{-\baselineskip}]
  \item For every goal \calltwo{g}{c_i} in the forest $\mf(q,c_q)$, the set
    \Cs is constraint-compact, where \Cs is the set of all the
    constraint stores $c_i$, projected and renamed w.r.t.\ the
    arguments of $g$.
%

  \item For every goal \calltwo{g}{c_g} in the forest $\mf(q,c_q)$, the set \As is
    constraint-compact, where \As is the set of all the answer
    constraints $c'$, projected and renamed w.r.t.\ the arguments of
    $g$, s.t. $c'$ is a successful derivation of \calltwo{g}{c_i} in the forest
    $\mf(q,c_q)$.
  \end{itemize}
\end{thee}

\begin{proof}
\cite{toman_theo_const_tabling} proves termination by observing
that the $SLG^C$ rewriting rules can be applied only finitely many
times.  We extend this proof to ensure that the TCLP
forest generated is finite and therefore the program execution
terminates.

\begin{enumerate}[wide=0.5em, leftmargin =*, nosep, before =
  \leavevmode\vspace{-\baselineskip}]
\item \label{item:literal} The execution can only generate a finite
  number of literals, up to variable renaming, because they are
  linearized (unifications take place in the constraints in the body)
  and the number of predicates in the program is finite.
\item \label{item:finiteII} The execution can only generate a finite
  number of TCLP forests $\bigsymbol{\tau}_P(g,c_g)$ because the
  number of possible literals is finite (point~\ref{item:literal}) and
  for each literal $g$, the set \Cs of its possible active constraint
  stores is constraint-compact. %
  That means that, for every subset of active constraint stores
  $C \sqsubseteq \Cs$, there exists a finite subset,
  $C_{fin} \subseteq C$ of possible most general calls, such that
  $\forall c \in C.\exists c' \in C_{fin}. c \sqsubseteq_{\domain}
  c'$. Therefore, at some point every new call will be entailed by
  some previous generator (this is checked in
  Def.~\ref{def:tclp-forest}.\ref{item:call_entail}).  %
\item \label{item:setII} The set of answers $Ans(g,c_g)$
  (Def.~\ref{def:tclp-forest}.\ref{item:ans_entail}) is finite because
  the set of possible most general answer
  constraints is finite.  The justification similar to that in point~\ref{item:finiteII}.
\item The number of children from a node resolved against clauses in
  $P$ (Def.~
  \ref{def:tclp-forest}.\ref{item:rest1}) is finite because the number of
  clauses in $P$ is finite.
\item The number of children from a node resolved by answer resolution
  (Def.~\ref{def:tclp-forest}.\ref{item:call_entail}) is finite
  because, by point~\ref{item:setII}, the set of answers $Ans(g,c_g)$
  is finite.
\end{enumerate}
\end{proof}

The intuition here is that for every subset $C$ from the set of all
possible constraint stores \Cs that can be generated when evaluating
a call to $P$, if there is
a finite subset $C_{fin} \subseteq C$ that covers (i.e., is as general
as) $C$, then, at some point, any call will be entailed by previous
calls, thereby allowing its suspension to avoid loops. %
Similarly, for every subset $A$ from the set of all possible answer
constraints \As that can be generated by a call, if there is a finite
subset $A_{fin} \subseteq A$ that covers $A$, then, at some point, any
answer will be entailed by a previous one, ensuring that the class of
answers $Ans(g,c_g)$ which entail any other possible answer returned
by the program is finite.\footnote{Note that a finite answer set does
  not imply a finite domain for the answers: the set of answers
  \code|Ans(q,c$_q$)=$\{$V>5$\}$| is finite, but the answer domain of
  \code|V| is infinite.}  Note that this result implies the classical
result that programs with the bounded depth term property always
finish under tabling with variant tabling, since the bounded depth
term property means that the number of possible constraints is finite
and therefore any constraint set covers itself.



\begin{exaa}
  The Herbrand domain (with constants and function symbols) and
  syntactic equality is not constraint-compact, and therefore
  termination of TCLP(\h) programs is not guaranteed.  However, in the
  case of programs which have only constants,
  the number of constraints that can be generated is finite, and
  therefore termination is ensured.
%
  Termination is also ensured (even with variant tabling) when a
  program can only generate terms with a bounded depth.  In this case,
  the number of distinct terms (and therefore of equality constraints)
  that can be generated is finite as well.
\end{exaa}

\begin{figure}
  \centering
  \begin{subfigure}[b]{.3\textwidth}
\begin{lstlisting}[style=MyProlog]
p(X) :-
    Y = f(X), 
    p(Y).
p(a).

\end{lstlisting}
    \caption{Program which finishes\newline under TCLP(\h).}
    \label{fig:terminate-p}
  \end{subfigure}
  \begin{subfigure}[b]{.3\textwidth}
\begin{lstlisting}[style=MyProlog]
nat(X) :-
    X.=.Y+1, 
    nat(Y).
nat(0).

\end{lstlisting}
    \caption{Natural numbers\\in TCLP(\Q).}
    \label{fig:terminate-nat}
  \end{subfigure}
  \begin{subfigure}[b]{.3\textwidth}
\begin{lstlisting}[style=MyProlog]
nat_k(X) :-
    X.=.Y+1, 
    nat_k(Y).
nat_k(0).
nat_k(X) :- X.>.1000.
\end{lstlisting}
    \caption{Describing infinitely many numbers in TCLP(\Q).}
    \label{fig:terminate-nat-k}
  \end{subfigure}
  \caption{TCLP programs under \h and \Q.}
\end{figure}

\begin{exaa}
  Fig.~\ref{fig:terminate-p} shows a program which loops
  in tabled Prolog and under \emph{variant} tabling.  The unification
  appears explicitly in the body for clarity.
  Although CLP(\h) is not constraint-compact, the constraints
  generated by that program under the query \code{?-p(X)} can make
  it finish.  Let examine its behavior from two points of view:

  \begin{description}
    [leftmargin=.7cm,labelindent=.7cm]
  \item[Compactness of the call constraint stores] %
    The set of all the constraint stores %
    generated for the predicate \code{p/1} under the query
    \callcode{p(X)}{true} is
    $C_{\mt{p(V)}} = \{\mt{true,\, V=f(X),\, V=f(f(X)),\,
      \dots}\}$.\footnote{The syntax $C_{\mt{p(V)}}$ means that (i) we
      are projecting all the calls to predicate \code|p/1| on the
      variables that call, and (ii) we are renaming these variables to
      be \code|V| in all the calls.  We could associate with every
      constraint store the names of the variables in the call in order
      to be able to compare different constraints stores (which is
      unnecessary after projection if there is only one variable in
      the call, but it would be needed if more than one variable is
      involved).  In order to avoid such an overload, and without loss
      of generality, we preferred to project and rename to a unique
      set of variables.}  It is constraint-compact because for every
    subset $C$ there is a finite set, e.g.  $C_{fin} = \{\mt{true}\}$,
    that covers $C$. %
  \item[Compactness of the answer constraints] %
    Additionally, the set of all answer constraints for the query,
    $A_{\text{\callcode{p(V)}{true}}} = \{\mt{V = a}\}$, is also
    constraint-compact because it is finite. Since both are
    constraint-compact, the execution terminates.%
    
  \item[Suspension due to call entailment] The first recursive call is
    \callcode{p(Y$_1$)}{Y$_1$=f(X)} and its projected and renamed constraint
    store is entailed by the initial store: \code{V=f(X)}
    $\sqsubseteq$ \code{true}.
    Therefore, TCLP evaluation suspends the recursive call, shifts
    execution to the second clause, and generates the answer \code{X=a}.
    This answer is given to the suspended recursive call, results in
    the inconsistent constraint store \code{Y$_1$=f(X)$\ \land\ $Y$_1$=a}, and the
    execution terminates.%
  \end{description}

\end{exaa}

\begin{exaa}
  Using the previous example (Fig.~\ref{fig:terminate-p}) under the
  query \code{?-p(a)}, the set of all the generated constraint stores
  is $C_{\mt{p(V)}} = \{\mt{V=a,\, V=f(a),\,V=f(f(a)),\,\dots}\}$. It
  is not constraint-compact and the execution does not terminate. Let
  us examine its behavior:
  \\[-1.5em]
    \begin{description}
      [leftmargin=.7cm,labelindent=.7cm]
    \item[The call constraint stores are not compact] The first recursive call is %
      \callcode{p(Y$_1$)}{X=a $\ \land\ $ Y$_1$=f(X)} %
      and the projection of its constraint store, \code{Y$_1$=f(a)},
      is not entailed by the initial one after renaming:
      \code{V=f(a)} $\not\sqsubseteq$ \code{V=a}. Then this call is evaluated
      and produces the second recursive call, %
      \callcode{p(Y$_2$)}{X=a$\ \land\ $Y$_1$=f(X)$\ \land\ $Y$_2$=f(f(X))}.  Its
      projected constraint store, \code{Y$_2$=f(f(a))}, is not
      entailed by any of the previous constraint stores, and so on
      with the rest of the recursive calls.  Therefore, the evaluation
      loops without terminating.
    \end{description}
  \end{exaa}

  Let us show the termination properties of the examples used 
   in~\cite{TCLP-tplp2019}. These examples show under what
  conditions programs would terminate even if the constraint domain
  is not constraint-compact.

  \begin{exaa}
    Fig.~\ref{fig:terminate-nat} shows a program which
    generates all the natural numbers using TCLP(\Q). Although CLP(\Q)
    is not constraint-compact, the constraint stores generated by that
    program for the query \code{?-X.<.10,nat(X)} are
    constraint-compact and the program finitely finishes. Let us look at
    its behavior from two points of view:
    \\[-1.5em]
    \begin{description}
    [leftmargin=.7cm,labelindent=.7cm]
  \item[Compactness of the call constraint stores and answer
    constraints] %
    The set of all constraint stores generated for the predicate
    \code{nat/1} under the query \callcode{nat(X)}{X<10} is
    $C_{\mt{nat(V)}} = \{\mt{V<10,\, V<9,\, \dots,\, V < -1, V < -2,\,
      \dots}\} $. It is constraint-compact because every subset
    $C \in C_{\mt{nat(V)}}$ is covered by $C_{fin} = \{\mt{V<10}\}$. %
    The set of all possible answer constraints for the query,
    $A_{\text{\callcode{nat(V)}{V<10}}} = \{\mt{V=0,\, \dots,\,
      V=9}\}$, is also constraint-compact because it is
    finite. Therefore, the program terminates.
  \item[Suspension due to call entailment] The first recursive call is
    \callcode{nat(Y$_1$)}{X<10$\ \land\ ${X=Y$_1$+1}} and the projection
    of its constraint store after renaming is entailed by the initial
    one since \code{V<9} $\sqsubseteq$ \code{V<10}.  Therefore, TCLP
    evaluation suspends in the recursive call, shifts execution to the
    second clause and generates the answer \code{X=0}. This answer is
    given to the recursive call, which was suspended, produces the
    constraint store \code{X<10$\ \land\
      $X=Y$_1$+1$\ \land\ $Y$_1$=0}, and generates the answer
    \code{X=1}. Each new answer \code{X$_n$=n} is used to feed the
    recursive call. When the answer \code{X=9} is given, it results in
    the (inconsistent) constraint store
    \code{X<10$\ \land\ $X=Y$_1$+1$\ \land\ $Y$_1$=9} and the
    execution terminates.
  \end{description}

\end{exaa}


\begin{exaa}
  The program in Fig.~\ref{fig:terminate-nat} does not
  terminate
  for the query \code{?-X.>.0,X.<.10,nat(X)}. Let us examine its
  behaviour: %
  \\[-1.5em]
  \begin{description}
    [leftmargin=.7cm,labelindent=.7cm]
  \item[The call constraint stores are not compact] The set of all
    constraint stores generated by the query %
    \callcode{nat(X)}{X>0$\ \land\ $X<10} is
    $C_{\mt{nat(V)}} = \{\mt{V>0 \land V<10},\,
    \mt{V>-1 \land V<9},\, \dots,\,
    \mt{V>-n \land V<(10-n)},\, \dots\}$, which it is not
    constraint-compact.  Note that \code{V} is, in successive calls,
    restricted to a sliding interval \code{[k,k+10]} which starts at
    \code{k=0} and decreases \code{k} in each recursive call.  No
    finite set of intervals can cover any subset of the possible
    intervals.

  \item[The evaluation loops] The first recursive call is
    \callcode{nat(Y$_1$)}{X>0$\ \land\ $X<10$\ \land\ $X=Y$_1$+1} and the
    projection of its constraint store is not entailed by the initial
    one after renaming since
    \code{(V>-1$\ \land\ $V<9)} $\not\sqsubseteq$ \code{(X>0$\ \land\ $X<10)}. Then
    this call is evaluated and produces the second recursive call,
    \callcode{nat(Y$_2$)}{X>0$\ \land\ $X<10$\ \land\ $X=Y$_1$+1$\ \land\ $Y$_1$=Y$_2$+1}.
    Again, the projection of its constraint store,
    \code{Y$_2$>-2$\ \land\ $Y$_2$<8}, is not entailed by any of the
    previous constraint stores, and so on.  The evaluation therefore
    loops.
  \end{description}
\end{exaa}

  \begin{exaa}\label{exa:termination-1}
    The program in Fig.~\ref{fig:terminate-nat} does not terminate
    with the query
    \code{?-nat(X)}.  
    \\[-1.5em]
    \begin{description}
      [leftmargin=.7cm,labelindent=.7cm]
    \item[Compactness of the call constraints stores] %
      The set of all constraint stores generated by the query
      \callcode{nat(X)}{true} is $C_{\mt{nat(V)}} =
      \{\mt{true}\}$. The set $C_{\mt{nat(V)}}$ is constraint-compact
      because it is finite.%


    \item[The answer constraints are not compact] However, the answer
      constraint set
      $A_{\text{\callcode{nat(V)}{true}}} = \{\mt{V=0,\, V=1,\,
        \dots,\, V=n,\, \dots}\}$ is not constraint-compact, and
      therefore the program does not terminate.
    \item[The evaluation does not terminate] %
      The first recursive call is \callcode{nat(Y$_1$)}{X=Y$_1$+1} and the
      projection of its constraint store\footnote{The equation in the
        body of the clause \code|X=Y$_1$+1| defines a relation between the
        variables but, since the domain of \code|X| is not restricted, its
        projection onto \code|Y$_1$| returns no constraints (i.e.,
        \code|Proj(Y$_1$, X=Y$_1$+1) = true|).} is entailed by the initial
      store. Therefore, the TCLP evaluation suspends the recursive
      call, shifts execution to the second clause, and generates the
      answer \code{X=0}.  This answer is used to feed the suspended
      recursive call, resulting in the constraint store
      \code{X=Y$_1$+1$\ \land\ $Y$_1$=0} which generates the answer \code{X=1}. Each
      new answer \code{X=n} is used to feed the suspended recursive call.
      Since the projection of the constraint stores on the call
      variables is \code{true}, the execution tries to generate infinitely
      many natural numbers.
    \end{description}

  \end{exaa}

  \begin{exaa}
    \label{exa:nat-k}
      Unlike what happens in pure Prolog/variant
      tabling, adding new clauses to a program under TCLP can make it
      terminate.\footnote{This depends on the strategy used by the
        TCLP engine to resume suspended goals.  An implementation that
        gathers all the answers for goals that can produce results
        first, and then these answers are used to feed suspended
        goals,  makes the exploration of the forests proceed in a
        breadth-first fashion.} As an example,
      Fig.~\ref{fig:terminate-nat-k} is the same as
      Fig.~\ref{fig:terminate-nat} with the addition of
      the clause \code{nat_k(X):-X.>.1000}.  Let us examine its behavior
      under the query \code{?-nat_k(X)}:
   \begin{description}
      [leftmargin=.7cm,labelindent=.7cm]
    \item[Compactness of call/answer constraint stores] %
      The set of all constraint stores generated remains
      $C_{\mt{nat\_k(V)}} = \{\mt{true}\}$. But the new clause makes  the
      answer constraint set become
      $A_{\text{\callcode{nat_k(V)}{true}}} = \{\mt{V=0},\, \mt{V=1},\, \dots,\, \mt{V=n},\,
      \dots,\, \mt{V>1000},\, \mt{V>1001},\, \dots,\, \mt{V>n},\, \dots\}$, which is
      constraint-compact because a constraint of the form \code{V>n}
      entails infinitely many constraints, i.e.\ it covers the infinite
      set \code|{V=n+1,$\dots$,V>n+1,$\dots$}|. Therefore, since
      both sets are constraint-compact, the program terminates.
      
    \item[First search, then consume] %
      The first recursive call \callcode{nat_k(Y$_1$)}{X = Y$_1$+1} is
      suspended and the TCLP evaluation shifts to the second clause
      which generates the answer \code{X=0}.  Then, instead of feeding
      the suspended call, the evaluation continues the search and
      shifts to the added clause, \code{nat_k(X):-X.>.1000}, and
      generates the answer \code{X>1000}.  Since no more clauses
      remain to be explored, the answer \code{X=0} is used, generating
      \code{X=1}. Then \code{X>1000} is used, resulting in the
      constraint store \code{X=Y$_1$+1$\ \land\ $Y$_1$>1000}, which
      generates the answer \code{X>1001}. %
      However, 
      \code{X>1001} is
      discarded because \code{X>1001} $\sqsubseteq$
      \code{X>1000}. Then, one by one each answer \code{X=n} is used,
      generating \code{X=n+1}. But when the answer \code{X=1000} is
      used, the resulting answer \code{X=1001} is discarded 
      because \code{X=1001} $\sqsubseteq$
      \code{X>1000}. At this point the evaluation terminates because
      there are no more answers to be consumed. The resulting set of
      answers is %
      \code|Ans(nat_k(X),true) = {X=0,X>1000,X=1,$\dots$,X=1000}|.
    \end{description}
    
  \end{exaa}

\section{The Role of Projection in TCLP}
\label{sec:import-prec-impl}

The detection of more particular calls and answers is
performed 
by checking entailment of the current constraint store of calls
(resp., answers) against the projected constraint store of a previous
call.  Some previous frameworks~\cite{tchr2007,bao2000} did not
implement a precise projection due to performance and implementation
issues.  Given that in some cases approximate projections can be more
efficient and/or easier to implement, it is worth exploring how
relaxing projection impacts soundness and completeness.
%
Let $c$ be a constraint store and let $c_s$ be a projection of $c$ on
some set of variables $S$.\footnote{In all cases
  the projected constraint store $c_S$ only has the variables in $S$ in
  common with the original store $c$.}
Let us also recall (Def.~\ref{def:valuation}) that a valuation is a
mapping from variables to domain constants and that a solution for a
constraint is a valuation that is consistent with the interpretation
of the constraint in its domain.  
We distinguish three possible projection variants:

\begin{description}[noitemsep]
\item [Precise projection (\textmd{denoted} $c \equiv c_s$)] $c_s$ is a projection of
  $c$ over some set of variables $S$, as defined in Def.~\ref{def:projection}.
\item[Over-approximating projection (\textmd{denoted}  $c \sqsubseteq c_s$)] The
  projected constraint $c_s$ is more general than the precise
  projection, e.g., some solutions for $c_s$ are not partial solutions
  for $c$.
%
  Any solution for $c$ is still a solution for $c_s$.
\item[Under-approximating projection (\textmd{denoted} $c \sqsupseteq c_s$)] $c_s$ is
  less general than the precise projection, e.g., there may be
  solutions for $c$ that are not solutions for $c_s$.
  Any solution of $c_s$ is still a (partial) solution for $c$.
\end{description}

Let us explain how these projection variants interact with the three
phases of the operational semantics described in
Section~\ref{sec:operational}:

\begin{itemize}
\item During the call entailment check (see
  Def.~\ref{def:tclp-forest}.\ref{item:call_entail}), if a new goal
  \calltwo{t}{c}, where $t$ is a tabled literal, does not entail a
  previous generator then, a new TCLP forest $\mf_P(t, c_s)$ is
  created and \calltwo{t}{c_s} is a new generator, where
  $c_s = Proj(vars(t),c)$.  Therefore, depending on the projection
  variant used, we have that:

  \begin{itemize}
  \item Using a precise projection, as already shown, the
    evaluation of the generator \calltwo{t}{c_s} would generate the same 
    answers as the evaluation of the goal \calltwo{t}{c}.
  \item Using an over-approximating projection, the generator
    \calltwo{t}{c_s} is more general than \calltwo{t}{c}, and
    therefore the evaluation of \calltwo{t}{c_s} may generate answers
    that are not consistent with the constraint store $c$. Note,
    however, that these answers will be filtered: when they are
    recovered and applied to a consumer (or to their generator)
    they will be checked for consistency against the constraint store
    of the call for which they are used.
  \item Using an under-approximating projection, the generator
    \calltwo{t}{c_s} is more particular than the goal \calltwo{t}{c},
    and, therefore, its evaluation may not generate answers that
    \calltwo{t}{c} would. Note that all of them would be consistent
    with $c$.
  \end{itemize}

  On the other hand, if a new goal \calltwo{t}{c'} entails a previous
  generator \calltwo{t}{c_s}, the goal \calltwo{t}{c'} is as usual
  marked as a consumer and would consume the answers generated by
  \calltwo{t}{c_s}.
  %

\item During the answer entailment check
  (Def.~\ref{def:tclp-forest}.\ref{item:ans_entail}), the final
  constraint store $a$ of each successful derivation of the evaluation
  of a generator \calltwo{t}{c_s} is projected to obtain the answer
  constraint $a_s$, i.e., $a_s = Proj(vars(t),a)$. Depending on the
  projection variant used we have that:

  \begin{itemize}
  \item Using a precise projection (denoted $a \equiv a_s$), as
    already proved, the resulting set of answer constraints for a
    generator does not add or exclude any valuation w.r.t.\ the set of
    its final constraint stores.
  \item Using an over-approximating projection (denoted
    $a \sqsubseteq a_s$), 
    the projected answer constraint $a_s$ may add valuations that are
    not consistent with the final constraint store $a$.
  \item Using an under-approximating projection (denoted
 $a_s \sqsubseteq a$), 
    $a_s$ may exclude valuations that are contained in the 
    constraint store $a$.
  \end{itemize}

\item During the application of the answers
  (Def.~\ref{def:tclp-forest}.\ref{item:rest2}), each answer constraint
  $a_s$ obtained during the evaluation of a generator is added to the
  constraint store $c$ of the goal that created the generator and the
  goals that were marked as consumers of that generator. If $a_s$ is
  consistent with $c$, i.e., $\domain \vDash c \land a_s$ the
  evaluation continues under the constraint store $c \land a_s$.
  Otherwise, it fails and the next answer constraint is retrieved.
\end{itemize}


We will now summarize how using non-precise projections impacts the
soundness and completeness of TCLP.
%
%
Tables~\ref{tab:ss} and~\ref{tab:cc} summarize whether soundness and
completeness (resp.) are preserved when using over- and
under-approximations for the projections in the call (column) and
answer (row) entailment check:
%
%
`\ok' in a location of each table means that the corresponding
combination of projection variants
preserves soundness (resp., completeness), while
%
`\no' means the opposite.
%
As expected, some combinations do not preserve soundness /
completeness.  Let us give an intuition behind these tables.


\begin{itemize}
\item In the top row of Table~\ref{tab:ss}, the only combination that
  may be unsound is the one that uses an over-approximation for the
  call projection: the answers may be more general than what a precise
  approximation would produce.  However, as mentioned before, when an
  answer is applied to a goal, a conjunction with the call constraint
  of that goal is made.  That balances the use of an
  over-approximation in the call.  This is in fact similar to the case
  of a consumer that uses answers from a more general generator.

\item The combinations in the middle row of Table~\ref{tab:ss} are not
  sound because over-approximations can produce answer constraints
  that allows for more valuations than a correct solution.

\item The cases in the bottom row of Table~\ref{tab:ss} are clearly
  sound as the projection of the answer constraints is more
  restrictive than a precise projection, and therefore it cannot
  introduce unwanted solutions.

\item The combinations in the rightmost column and the bottom-most row
  of Table~\ref{tab:cc} may not be complete because they either
  restrict the projected store for a call or they restrict the
  answers.  In both cases, solutions may be missed.

\item The rest of the cases in Table~\ref{tab:cc} may use projections
  more relaxed than a precise one, so additional solutions can be
  generated, but no solution should be removed.
\end{itemize}

\begin{table}
  \caption{Combinations of precise, over- and under- approximation
    (`$\equiv$', `$\sqsubseteq$' and `$\sqsupseteq$') \\for the call
    and answer entailment check.}

  \begin{subtable}{.43\linewidth}
    \caption{Soundness preservation.}
    \label{tab:ss}
    \centering
    \begin{tabular}{p{1.2cm}|ccc}
      \toprule
        & $c \equiv c_s$ & $c \sqsubseteq c_s$ & $c \sqsupseteq c_s$ \\[1em]
      \midrule
      $a \equiv a_s$          & \ok &  \ok &  \ok \\[1em]
      $a \sqsubseteq a_s$     & \no &  \no &  \no \\[1em]
      $a \sqsupseteq a_s$     & \ok &  \ok &  \ok \\[1em]
      \bottomrule    
    \end{tabular}
  \end{subtable}
  \begin{subtable}{.05\linewidth}
    ~
  \end{subtable}
  \begin{subtable}{.43\linewidth}
    \caption{Completeness preservation.}
    \label{tab:cc}
    \centering
    \begin{tabular}{p{1.2cm}|ccc}
      \toprule
        & $c \equiv c_s$ & $c \sqsubseteq c_s$ & $c \sqsupseteq c_s$ \\[1em]
      \midrule
      $a \equiv a_s$          & \ok &  \ok &  \no \\[1em]
      $a \sqsubseteq a_s$     & \ok &  \ok &  \no \\[1em]
      $a \sqsupseteq a_s$     & \no &  \no &  \no \\[1em]
      \bottomrule    
    \end{tabular}
  \end{subtable}
\end{table}

%
%
Some approximate projections can be more efficient and/or easier to
implement than precise projections, and that justifies their use in
specific scenarios.
For brevity, let us comment on the combinations that preserve
soundness and completeness, $\equiv / \equiv$ and
$\sqsubseteq / \equiv$, and a combination that over-approximates the
answers while using a precise projection in the calls,
$\equiv / \sqsubseteq$:
%

\begin{itemize}

\item \ensuremath{\equiv / \equiv}:
  Precise projection `$\equiv$' in the call and answer entailment
  check. This is optimal in the sense that it guarantees soundness
  and completeness, removes redundant answers, and reduces the search
  space. It has been used 
  in~\cite{TCLP-tplp2019}.

\item $\sqsubseteq / \equiv$:
  Over-approximate projection `$\sqsubseteq$' for the calls and
  precise projection `$\equiv$' for the answers. 
  In this case, generators may generate answers that a precise
  projection would not, since
  they start with a more relaxed constraint store (which can turn
  terminating queries into non-terminating ones).  
  This of course preserves completeness.
  Soundness is preserved because answer constraints that are not
  consistent with the initial goal constraint store $c$ will be discarded.

\begin{exaa}
  \leftskip=0pt\parindent=0pt%
  ~\\
  Call abstraction~\cite{tchr2007} is an extreme example, where the
  constraint store associated with the tabled call is not taken into
  account for the execution of the call (i.e., the projection of a
  constraint store is always the constraint \code{true}). Therefore, a
  generator with \code{true} as constraint store will be entailed by
  any subsequent call because $c \sqsubseteq true$ for any constraint
  $c$. As mentioned above (see Example~\ref{exa:termination-1}), this
  loses several benefits of tabling with constraints because we have
  to compute all the possible results for an unrestricted call and
  then filter them through the constraint store active at call-time.
  However, soundness is preserved.
\end{exaa}

\item $\equiv / \sqsubseteq$:
  Precise projection `$\equiv$' for the calls and over-approximate
  projection `$\sqsubseteq$' for the answers.  This combination is
  relevant because applications such as program analyzers based on
  abstract interpretation can be seen as performing an execution in an
  abstract domain that over-approximates the values of the concrete
  domain to guarantee
  termination. 
%
%
  This over-approximation can be implemented with a constraint system
  that reflects the operations of abstract domain and whose answer
  projections are as well over-approximated.
  %
  Such an over-approximation can increase performance because a more
  general answer would be more frequently entailed
  by other answers, reducing the number of answers stored and the
  number of resumptions.
  
  However, using an over-approximation in the answer projections may
  make answer resolution to lose precision arbitrarily.
  When an answer constraint $a$ for a generator \calltwo{t}{c_s} is
  projected to obtain the over-approximated answer constraint $a_s$,
  this answer is saved in case it can be reused later on.

  When a (more concrete) consumer \calltwo{t}{c'} performs answer
  resolution consuming $a_s$, the resulting answer would be
  $c' \land a_s$.  Depending on how the over-approximation is
  performed, $c' \land a_s$ can be arbitrarily less precise (or even
  incomparable) than what would have been the result of executing
  $(t, c')$ against program clauses and then abstracting it.  However,
  there are some cases where by putting some conditions on when an
  answer is reused, this problem can be worked around.

  \begin{exaa}
    \leftskip=0pt\parindent=0pt%
    ~\\
    The implementation of PLAI with TCLP presented
    in~\cite{arias19:ciaopp-tclp} is an example of this option. In
    that paper, an abstract interpreter is built using TCLP where the
    abstract domain and its operations are modeled using a constraint
    system.  One of these computes the lowest upper bound of different
    abstract substitutions resulting from the analysis of each clause
    of a predicate, to return the abstract substitution corresponding
    to the predicate.  If $a_1$ and $a_2$ are the abstract
    substitutions at the end of the bodies of two (normalized) clauses
    $p_1$ and $p_2$, one would like to calculate
    $Proj(var(p), a_1 \lor a_2)$, where $Proj$ may be an
    overapproximation.  When answer substitutions for each clause are
    projected and stored separately, composing them is done by
    computing $Proj(var(p),a_1) \sqcup Proj(var(p),a_2)$, which can be
    less precise than $Proj(var(p), a_1 \lor a_2)$.
    %
    %
    That makes the predicate-level abstract substitution for $p$ to
    possibly be an overapproximation of the more precise abstract
    version.

    The tabled abstract substitution for goal $p$ can be retrieved and
    used to compute the exit substitution for another goal $p'$ when
    $p' \sqsubseteq p$, using answer resolution. In that case, the
    exit substitution for $p'$ can be arbitrarily less precise than
    what would have been obtained by analyzing directly $p'$ using
    clause resolution and then abstracting.  We worked around this
    issue by reusing substitutions only in the case that $p$ and $p'$
    correspond to the same point in the lattice, i.e., when their
    entry substitutions are (semantically) equal modulo variable
    renaming.  This ensures that the abstract substitution for $p$ can
    be used for $p'$ without incurring in additional loss of
    precision, because the analysis results for $p'$ and $p$ should be
    the same.
         
  %

  \end{exaa}

\end{itemize}

\noindent
To the best of our knowledge, there are no examples where
under-approximate projections `$\sqsupseteq$' are used.  However,
since they preserve soundness (except when an over-approximation is
used for answer projection, which is neither sound not complete), they
can be useful in scenarios where the existence of a solution is enough
to answer a question.  This would the case, for example, for program
verification: a solution for a query to a TCLP program that uses
underapproximations and looks for counterexamples to the correctness
of a program would demonstrate the existence of an error in the
program, even if the answer only shows a subset of the domain of the
variables for which the program exhibits a wrong behavior.

%
%

\section{Conclusions}

We have extended the theoretical basis of tabled constraint logic
programming for a top-down execution. We have characterized the
properties that the constraint solver should holds in order to
guarantee soundness and completeness. For non constraint-compact
constraint systems, we define sufficient conditions for queries to
terminate.
For constraint domains without a precise implementation of the
projection of constraint stores, we evaluate how relaxing the
projection impacts soundness, completeness, and termination.

From our point of view, the new formalization in terms of soundness,
completeness and termination would facilitate the implementation of
new tabled constraint logic programming systems and their integration
with a lager number of constraint domain (e.g., constraint solvers over
finite domains).

\bibliographystyle{acmtrans}

\begin{thebibliography}{}

\bibitem[\protect\citeauthoryear{Arias}{Arias}{2016}]{arias-dc-iclp2016}
{\sc Arias, J.} 2016.
\newblock {T}abled {CLP} for {R}easoning over {S}tream {D}ata.
\newblock In {\em {Technical Communications of the 32nd Int'l.\ Conference on
  Logic Programming}}. Vol.~52. OASIcs, 1--8.
\newblock Doctoral Consortium.

\bibitem[\protect\citeauthoryear{Arias and Carro}{Arias and
  Carro}{2019a}]{TCLP-tplp2019}
{\sc Arias, J.} {\sc and} {\sc Carro, M.} 2019a.
\newblock {D}escription, {I}mplementation, and {E}valuation of a {G}eneric
  {D}esign for {T}abled {CLP}.
\newblock {\em Theory and Practice of Logic Programming\/}~{\em 19,\/}~3 (May),
  412--448.

\bibitem[\protect\citeauthoryear{Arias and Carro}{Arias and
  Carro}{2019b}]{arias19:ciaopp-tclp}
{\sc Arias, J.} {\sc and} {\sc Carro, M.} 2019b.
\newblock {Evaluation of the Implementation of an {A}bstract Interpretation
  Algorithm using Tabled CLP}.
\newblock {\em Theory and Practice of Logic Programming\/}~{\em 19,\/}~5-6
  (September), 1107--1123.
\newblock Special Issue on ICLP'19.

\bibitem[\protect\citeauthoryear{Arias and Carro}{Arias and
  Carro}{2019c}]{atclp-padl2019}
{\sc Arias, J.} {\sc and} {\sc Carro, M.} 2019c.
\newblock Incremental evaluation of lattice-based aggregates in logic
  programming using modular {TCLP}.
\newblock In {\em 21st Int'l. Symposium on Practical Aspects of Declarative
  Languages}, {J.~J. Alferes} {and} {M.~Johansson}, Eds. LNCS, vol. 11372.
  Springer, 98--114.

\bibitem[\protect\citeauthoryear{Charatonik, Mukhopadhyay, and
  Podelski}{Charatonik
  et~al\mbox{.}}{2002}]{Charatonik02:model_checking_tabulation}
{\sc Charatonik, W.}, {\sc Mukhopadhyay, S.}, {\sc and} {\sc Podelski, A.}
  2002.
\newblock {Constraint-Based Infinite Model Checking and Tabulation for
  Stratified CLP}.
\newblock In {\em ICLP'02}, {P.~J. Stuckey}, Ed. Lecture Notes in Computer
  Science, vol. 2401. Springer, 115--129.

\bibitem[\protect\citeauthoryear{{Chico~de~Guzm\'{a}n}, Carro, Hermenegildo,
  and Stuckey}{{Chico~de~Guzm\'{a}n}
  et~al\mbox{.}}{2012}]{chico-tclp-flops2012-large}
{\sc {Chico~de~Guzm\'{a}n}, P.}, {\sc Carro, M.}, {\sc Hermenegildo, M.~V.},
  {\sc and} {\sc Stuckey, P.} 2012.
\newblock {A} {G}eneral {I}mplementation {F}ramework for {T}abled {CLP}.
\newblock In {\em 15th Int'l.\ Symposium on Functional and Logic Programming},
  {T.~Schrijvers} {and} {P.~Thiemann}, Eds. LNCS, vol. 7294. Springer Verlag,
  104--119.

\bibitem[\protect\citeauthoryear{Cui and Warren}{Cui and
  Warren}{2000}]{bao2000}
{\sc Cui, B.} {\sc and} {\sc Warren, D.~S.} 2000.
\newblock {A} system for {T}abled {C}onstraint {L}ogic {P}rogramming.
\newblock In {\em Int'l.\ Conference on Computational Logic}. LNCS, vol. 1861.
  Springer, 478--492.

\bibitem[\protect\citeauthoryear{Dawson, Ramakrishnan, and Warren}{Dawson
  et~al\mbox{.}}{1996}]{Dawson:pldi96}
{\sc Dawson, S.}, {\sc Ramakrishnan, C.~R.}, {\sc and} {\sc Warren, D.~S.}
  1996.
\newblock Practical {P}rogram {A}nalysis {U}sing {G}eneral {P}urpose {L}ogic
  {P}rogramming {S}ystems -- {A} {C}ase {S}tudy.
\newblock In {\em Proceedings of the ACM SIGPLAN'96 Conference on Programming
  Language Design and Implementation}. ACM Press, New York, USA, 117--126.

\bibitem[\protect\citeauthoryear{Falaschi, Levi, Martelli, and
  Palamidessi}{Falaschi et~al\mbox{.}}{1989}]{falaschi-vars89}
{\sc Falaschi, M.}, {\sc Levi, G.}, {\sc Martelli, M.}, {\sc and} {\sc
  Palamidessi, C.} 1989.
\newblock Declarative {M}odeling of the {O}perational {B}ehaviour of {L}ogic
  {P}rograms.
\newblock {\em Theoretical Computer Science\/}~{\em 69}, 289--318.

\bibitem[\protect\citeauthoryear{Gabbrielli and Levi}{Gabbrielli and
  Levi}{1991}]{gabbrielli1991modeling}
{\sc Gabbrielli, M.} {\sc and} {\sc Levi, G.} 1991.
\newblock {Modeling Answer Constraints in Constraint Logic Programs}.
\newblock In {\em Proc. 8th Int'l Conference on Logic Programming}. 238--252.

\bibitem[\protect\citeauthoryear{Gange, Navas, Schachte, S{\o}ndergaard, and
  Stuckey}{Gange et~al\mbox{.}}{2013}]{navas2013-FTCLP}
{\sc Gange, G.}, {\sc Navas, J.~A.}, {\sc Schachte, P.}, {\sc S{\o}ndergaard,
  H.}, {\sc and} {\sc Stuckey, P.~J.} 2013.
\newblock {Failure Tabled Constraint Logic Programming by Interpolation}.
\newblock {\em {TPLP}\/}~{\em 13,\/}~4-5, 593--607.

\bibitem[\protect\citeauthoryear{Jaffar and Maher}{Jaffar and
  Maher}{1994}]{survey94}
{\sc Jaffar, J.} {\sc and} {\sc Maher, M.} 1994.
\newblock {C}onstraint {L}ogic {P}rogramming: {A} {S}urvey.
\newblock {\em Journal of Logic Programming\/}~{\em 19/20}, 503--581.

\bibitem[\protect\citeauthoryear{Jaffar, Santosa, and Voicu}{Jaffar
  et~al\mbox{.}}{2004}]{jaffar04:timed_automata}
{\sc Jaffar, J.}, {\sc Santosa, A.~E.}, {\sc and} {\sc Voicu, R.} 2004.
\newblock {A {CLP} Proof Method for Timed Automata}.
\newblock In {\em RTSS}. IEEE Computer Society, 175--186.

\bibitem[\protect\citeauthoryear{Janssens and Sagonas}{Janssens and
  Sagonas}{1998}]{janssens98:tabling_for_AI-tapd}
{\sc Janssens, G.} {\sc and} {\sc Sagonas, K.} 1998.
\newblock {O}n the {U}se of {T}abling for {A}bstract {I}nterpretation: {A}n
  {E}xperiment with {A}bstract {E}quation {S}ystems.
\newblock In {\em Tabulation in Parsing and Deduction}.

\bibitem[\protect\citeauthoryear{Kanamori and Kawamura}{Kanamori and
  Kawamura}{1993}]{kanamori93:absint-oldt}
{\sc Kanamori, T.} {\sc and} {\sc Kawamura, T.} 1993.
\newblock {Abstract Interpretation Based on OLDT Resolution}.
\newblock {\em Journal of Logic Programming\/}~{\em 15}, 1--30.

\bibitem[\protect\citeauthoryear{Ramakrishna, Ramakrishnan, Ramakrishnan,
  Smolka, Swift, and Warren}{Ramakrishna
  et~al\mbox{.}}{1997}]{ramakrishna97:model_checking_tabling}
{\sc Ramakrishna, Y.}, {\sc Ramakrishnan, C.}, {\sc Ramakrishnan, I.}, {\sc
  Smolka, S.}, {\sc Swift, T.}, {\sc and} {\sc Warren, D.} 1997.
\newblock {E}fficient {M}odel {C}hecking {U}sing {T}abled {R}esolution.
\newblock In {\em Computer Aided Verification}. LNCS, vol. 1254. Springer
  Verlag, 143--154.

\bibitem[\protect\citeauthoryear{Revesz}{Revesz}{1993}]{revesz1993closed}
{\sc Revesz, P.~Z.} 1993.
\newblock {A Closed-Form Evaluation for Datalog Queries with Integer
  (Gap)-Order Constraints}.
\newblock {\em Theoretical Computer Science\/}~{\em 116,\/}~1, 117--149.

\bibitem[\protect\citeauthoryear{Schrijvers, Demoen, and Warren}{Schrijvers
  et~al\mbox{.}}{2008}]{tchr2007}
{\sc Schrijvers, T.}, {\sc Demoen, B.}, {\sc and} {\sc Warren, D.~S.} 2008.
\newblock {TCHR}: a {F}ramework for {T}abled {CLP}.
\newblock {\em Theory and Practice of Logic Programming\/}~4 (Jul), 491--526.

\bibitem[\protect\citeauthoryear{Swift and Warren}{Swift and
  Warren}{2010}]{swift2010:subsumption}
{\sc Swift, T.} {\sc and} {\sc Warren, D.~S.} 2010.
\newblock Tabling with answer subsumption: Implementation, applications and
  performance.
\newblock In {\em Logics in Artificial Intelligence}. Vol. 6341. 300--312.

\bibitem[\protect\citeauthoryear{Tamaki and Sato}{Tamaki and
  Sato}{1986}]{tamaki.iclp86}
{\sc Tamaki, H.} {\sc and} {\sc Sato, M.} 1986.
\newblock {OLD} {R}esolution with {T}abulation.
\newblock In {\em Third International Conference on Logic Programming}. Lecture
  Notes in Computer Science, Springer-Verlag, London, 84--98.

\bibitem[\protect\citeauthoryear{Toman}{Toman}{1997}]{toman_theo_const_tabling}
{\sc Toman, D.} 1997.
\newblock {M}emoing {E}valuation for {C}onstraint {E}xtensions of {D}atalog.
\newblock {\em Constraints\/}~{\em 2,\/}~3/4, 337--359.

\bibitem[\protect\citeauthoryear{van Emden and Kowalski}{van Emden and
  Kowalski}{1976}]{EmKo76}
{\sc van Emden, M.~H.} {\sc and} {\sc Kowalski, R.~A.} 1976.
\newblock {T}he {S}emantics of {P}redicate {L}ogic as a {P}rogramming
  {L}anguage.
\newblock {\em Journal of the {ACM}\/}~{\em 23}, 733--742.

\bibitem[\protect\citeauthoryear{Warren}{Warren}{1992}]{Warren92}
{\sc Warren, D.~S.} 1992.
\newblock {M}emoing for {L}ogic {P}rograms.
\newblock {\em Communications of the ACM\/}~{\em 35,\/}~3, 93--111.

\bibitem[\protect\citeauthoryear{Warren, Hermenegildo, and Debray}{Warren
  et~al\mbox{.}}{1988}]{pracabsin}
{\sc Warren, R.}, {\sc Hermenegildo, M.}, {\sc and} {\sc Debray, S.~K.} 1988.
\newblock {O}n the {P}racticality of {G}lobal {F}low {A}nalysis of {L}ogic
  {P}rograms.
\newblock In {\em Fifth International Conference and Symposium on Logic
  Programming}. {MIT} Press, 684--699.

\bibitem[\protect\citeauthoryear{Zou, Finin, and Chen}{Zou
  et~al\mbox{.}}{2005}]{zou05:owl-tabling}
{\sc Zou, Y.}, {\sc Finin, T.}, {\sc and} {\sc Chen, H.} 2005.
\newblock {F-OWL}: {A}n {I}nference {E}ngine for {S}emantic {W}eb.
\newblock In {\em Formal Approaches to Agent-Based Systems}. LNCS, vol. 3228.
  Springer Verlag, 238--248.

\end{thebibliography}
\label{sec:bibliography}

\end{document}